\documentclass[11pt]{article}

\usepackage[round]{natbib}

\usepackage{graphicx}

\usepackage{amsmath}
\usepackage{amsfonts}
\usepackage{enumerate}
\usepackage{lscape}
\usepackage{multirow}

\usepackage{amsthm}

\newtheorem{lemma}{Lemma}

\newtheorem{proposition}[lemma]{Proposition}

\usepackage{hyperref}

\usepackage{authblk}

\newcommand{\change}[1]{#1}
\newcommand{\newchange}[1]{#1}

\begin{document}

\title{Template-based Minor Embedding for \\ Adiabatic Quantum Optimization}

\author[1]{Thiago Serra} 
\author[2]{Teng Huang} 
\author[3]{Arvind Raghunathan}
\author[4]{David Bergman}

\affil[1]{Bucknell University}
\affil[2]{Lingnan (University) College, Sun Yat-sen University}
\affil[3]{Mitsubishi Electric Research Laboratories}
\affil[4]{University of Connecticut}

\date{}

\maketitle

\begin{abstract}
Quantum Annealing~(QA) can be used to quickly obtain near-optimal solutions for Quadratic Unconstrained Binary Optimization (QUBO) problems. 
In QA hardware, each decision variable of a QUBO should be mapped to one or more adjacent qubits in such a way that pairs of variables defining a quadratic term in the objective function are mapped to some pair of adjacent qubits.
However, qubits have limited connectivity in existing QA hardware.  
This has spurred work on preprocessing algorithms for embedding the graph representing problem variables with quadratic terms into the hardware graph representing qubits adjacencies, such as the Chimera graph in hardware produced by D-Wave Systems.
In this paper, we use integer linear programming to search for an embedding of the problem graph into certain classes of minors of the Chimera graph, which we call \emph{template embeddings}. One of these classes corresponds to complete bipartite graphs, for which we show the limitation of the existing approach based on minimum Odd Cycle Transversals~(OCTs). \newchange{One} of the formulations presented \newchange{is} exact, and thus can be used to certify the absence of a minor embedding \newchange{using that template}. On an extensive test set consisting of random graphs from five different classes of varying size and sparsity, we can embed  more graphs than a state-of-the-art OCT-based approach, \change{our  approach scales better with the hardware size, and the runtime is generally orders of magnitude smaller.} 
\end{abstract}

\section{Introduction}

Quantum Annealing~(QA) is a technique based on a system of quantum particles in which the Hamiltonian---an operator corresponding to the sum of potential and kinetic energies---is modeled after the objective function of a binary optimization problem~\citep{finnila1994quantum,kadowaki1998quantum}. 
Adiabatic Quantum Computation (AQC) is a particular form of QA 
that exploits the adiabatic theorem of quantum physics~\citep{Farhi2000,Farhi2001}, 
which states that a quantum system in the ground state evolves predominantly in the ground state---i.e., at the lowest possible energy level---if the rate at which the system changes is sufficiently small. 
Hence, AQC can approximately solve unconstrained optimization problems 
by slowly evolving a quantum system from a configuration in known 
ground state toward a configuration in which the ground state 
corresponds to a solution of minimum value for the optimization problem.

It is possible to model many combinatorial optimization problems with such a technique~\citep{Lucas2014}.  
However, it is not known if the quantum system can always evolve in polynomial time on the inputs in order to reach a quantum speedup---i.e., solve a problem much faster than any classical algorithm would~\citep{SpeedUp}. 
For example, 
there are cases in which a linear transformation between quantum states requires exponential time~\citep{VanDam2001}. 
Nevertheless, for some problems the quantum system evolves efficiently~\citep{Farhi2000} or at 
least faster than classical algorithms~\citep{Lucas2018}.     
When compared with commercial optimization software,  
AQC has been found to perform better in one case~\citep{McGeoch2013} and comparable in another~\citep{Coffrin}. 
\change{Since the solutions obtained through AQC are sometimes suboptimal, 
\cite{Dash} observes that AQC is arguably comparable to heuristics that can find optimal solutions with high probability. 
Hence, as this technology matures, its competitive advantage 
will depend on the speed of convergence as well as on being more reliable than such heuristics.
}

AQC has been used to approximately solve Quadratic Unconstrained Binary Optimization~(QUBO) problems of the form 
\[
\min_{x \in \{0,1\}^n} x^T Q x \qquad \change{\left( \text{or} \qquad \max_{x \in \{0,1\}^n} x^T Q x \right)\footnotemark},
\]
\footnotetext{\change{We regard the maximization version, which is preferred by some authors, as interchangeable with the minimization version by negating $Q$. For that reason, we only consider the minimization version for the rest of the paper.}}
where $Q \in \mathbb{R}^{n \times n}$ are input data and $x$ are decision variables.\footnote{According to \citet{boros2002pseudo}, 
QUBO is also known as a quadratic pseudo-Boolean function in the literature since~\cite{kalantari1986quadratic}.} %
Without loss of generality, 
let $Q$ be an upper-triangular matrix. 
In this paper, 
we will refer to any device implementing AQC as a \emph{quantum annealer}, 
or a QA hardware.

Although existing QA hardware may potentially be used on problems with as many as 2048 variables, 
in practice only problems that are much smaller or substantially sparse can be directly formulated. 
To circumvent this limitation to some extent, 
we may use a surrogate formulation in which the optimal solutions can be mapped to optimal solutions of the original problem. 
For defining such formulation, it is 
often necessary to analyze the structure of the 
QA hardware 
and 
of the optimization problem that we want to solve. 
\change{
That 
entails determining if the graph associated with the problem is a minor of the graph associated with the QA hardware,  
which is an NP-complete problem in general~\citep{Johnson87}. 
}

In this paper, we show that Integer Linear Programming~(ILP) can be effectively used for this preprocessing step, in which we determine how a QUBO problem can be modeled in QA hardware if the qubits have limited connectivity. 
We propose ILP formulations based on particular minors of the QA hardware, which we denote \emph{template embeddings}. 

We introduce the minor embedding problem for QA hardware and previous approaches in Section~\ref{sec:back}, and then contextualize our contribution in Section~\ref{sec:contr}. 
We describe the template embeddings, their properties, and corresponding ILP formulations in \newchange{Sections~\ref{sec:bte} and \ref{sec:qte}}. We present experimental results in Section~\ref{sec:exps} and final remarks in Section~\ref{sec:concl}.

\section{Background}\label{sec:back}

In this section we describe the graph embedding problem associated with QA hardware, 
in particular with Chimera graphs, 
and the previous approaches to this problem.

\subsection{Hardware and Problem Graphs}

The quantum annealer solves the Ising formulation
\[
\min\limits_{y \in \{-1,1\}^n} y^T J y + h^T y , 
\]
where $J \in \mathbb{R}^{n \times n}$ and $h \in \mathbb{R}^n$ are input data and 
$y$ are decision variables. 
Each binary variable $y_i$ is associated with a \emph{qubit} $i$, the basic unit of quantum information, 
and its value corresponds to the \emph{magnetic spin} of the quantum transistor that physically implements the qubit~\citep{DWaveWebsite}. 
The linear coefficient $h_i$ for each variable $y_i$ corresponds to the 
\emph{bias} of qubit $i$. The quadratic term $J_{i j}$ for each pair of variables 
$y_i$ and $y_j$, if nonzero, implies the existence of a \emph{coupler} between qubits $i$ and 
$j$, and the value of $J_{i j}$ represents the \emph{strength} of the coupler. 
For each QUBO formulation, there is a corresponding Ising formulation for which the variables 
with values -1 and 1 in an optimal solution of the Ising problem correspond to variables with 
values 0 and 1 in an optimal solution of QUBO, which is such that the zero elements in $Q$ are also zero in $J$~\citep{Choi2008}. 
Hence, we may assume that matrices $Q$ and $J$ have the same nonzero elements.

In practice, the hardware graph has a sparse structure. 
Due to engineering limitations, 
each qubit can only be coupled with a 
limited set of other qubits~\citep{Choi2008}. 
That implies that we can only directly solve problems with as many variables as the number of qubits if these problems follow the same sparsity structure as the QA hardware. 

A QA hardware can be modeled as an undirected graph in which the vertices correspond to qubits and the edges to 
pairs of coupled qubits. 
We denote it 
as the 
\emph{hardware graph} $H$. 
Similarly, 
we consider a \emph{problem graph} $G$ in which the vertices correspond to decision variables of an Ising formulation and the edges to 
pairs of variables with a quadratic term in the objective function. 
Let $G = (V(G), E(G))$,   
where $V(G) := \{ v_1, v_2, \ldots, v_n \}$ is the set of vertices of $G$ 
and $E(G)$ is a set of edges in which $\{v_i, v_j\} \in E(G)$ if $J_{i,j} \neq 0$. 
Similarly, let $H = (U(H), F(H))$, where vertex $u_i \in U(H)$ 
corresponds to qubit $i$ 
and $\{u_i, u_j\} \in F(H)$ implies that qubits $i$ and $j$ are coupled.
For convention, 
we will use $V_i$ for a subset of $V(G)$ and $U_i$ for a subset of vertices of $H$ or any of its minors.

A problem can be directly solved in QA hardware if there is a 
subgraph $H'$ of $H$ that is \emph{isomorphic} to $G$, 
i.e., there is a bijective mapping between the vertices of 
$G$ and $H'$ such that adjacent vertices in $G$ are mapped to adjacent vertices in $H'$. However, this 
greatly limits the class of problems that can be solved \change{if the vertices of $H$ have small degrees}.

The class of solvable problems can be enlarged by allowing each vertex of the problem graph $G$ to 
be mapped to possibly multiple vertices in the hardware graph $H$. This imposes the 
additional requirement that qubits associated with a  
vertex in $G$ have the same spin in the ground state.  
For example, two coupled qubits corresponding to vertices 
$u_i$ and $u_j$ can be induced to have the same spin 
in the ground state if $J_{i j}$ is negative and sufficiently large 
in absolute value~\citep{Kaminsky2004a,Kaminsky2004b}. 
In that case, those \emph{physical} qubits define a single 
\emph{logical} qubit. 
The multiplicity of physical qubits increases the number of 
neighbors and 
make it possible to embed a graph $G$ with higher 
connectivity than that present in $H$. 
More generally, we~say that $G$ can be embedded in a hardware graph $H$ if $G$ is a \emph{minor} of $H$~\citep{Choi2008}. 
A minor of a graph is any graph that can be obtained by a sequence of vertex and edge deletions as well as edge contractions~\citep{BondyMurty}. 
In the example above, contracting edge $\{u_i,u_j\}$ produces a graph in which vertices $u_i$ and $u_j$ are replaced by a vertex $u'$ that is adjacent to any vertex that was originally adjacent to either $u_i$ or $u_j$. 
Hence, the \emph{embedding} of $G$ in $H$ consists of assigning each vertex $v_i \in V(G)$ to a distinct set of vertices $U_i \subseteq U(H)$ 
such that the induced subgraph on $U_i$ is connected  
and, for each edge $\{v_i, v_j\} \in E(G)$, there exists $u_k \in U_i$ and $u_l \in U_j$ such that $(u_k, u_l) \in F(H)$. 

In the QAs that are produced and that are currently made commercially available by D-Wave Systems, 
the hardware graph follows the structure of a \emph{Chimera graph}. 
Let us denote it as a graph $C_{M,N,L}$ 
such that $ 2 M N L$ vertices 
are distributed in a grid of $M N$ cells. Each cell contains $2 L$ vertices.
Each of those cells is a complete bipartite graph $K_{L,L}$, where a left and a right partition each contain $L$ vertices. 
For ease of explanation, 
let us number the vertices in each partition from 1 to $L$. 
The $i$-th vertex in each left (right) partition is also adjacent to the corresponding $i$-th vertex of the left (right) partition in the cell above (to the left) and below (to the right). 
Figure~\ref{fig:Chimera} depicts $C_{2,2,4}$. 

The following hardware graphs have been used in QA hardware: 
$C_{4,4,4}$ in D-Wave One, $C_{8,8,4}$ in D-Wave Two, $C_{12,12,4}$ in D-Wave 2X, and $C_{16,16,4}$ in D-Wave 2000Q~\citep{Pegasus}. 
\change{In all of those cases, the maximum degree of a vertex is 6, 
which makes the problem of finding minor embeddings crucial to leverage such AQC hardware.} 
Since $M=N$ in all cases, 
we will follow the convention  of considering Chimera graphs of the form $C_{M,M,L}$. 

\begin{figure}[h!]
\centering
\includegraphics[width=0.5\columnwidth]{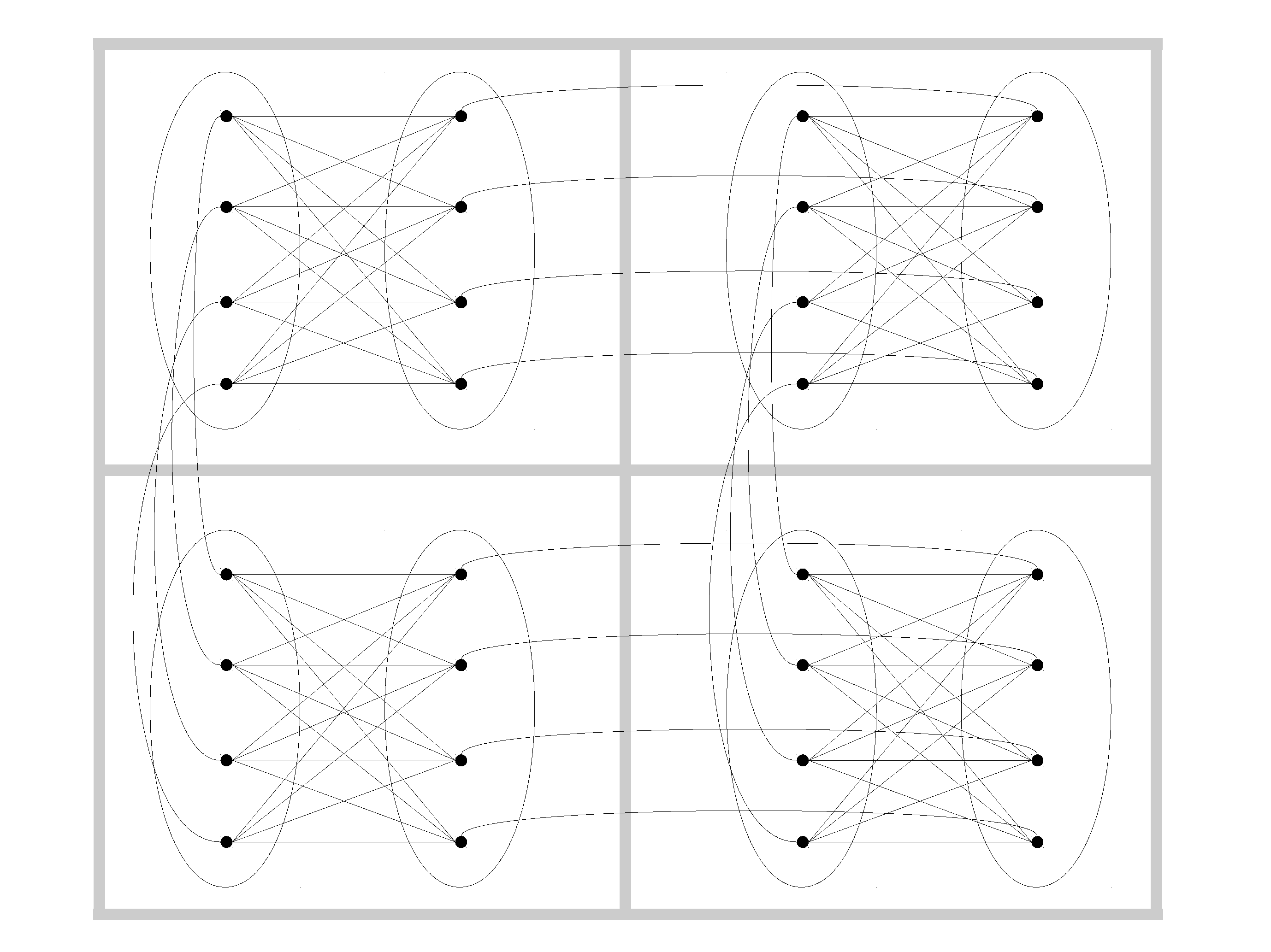}
\caption{Chimera graph $C_{2,2,4}$ 
with a $K_{4,4}$ on each cell of a $2 \times 2$ grid.}
\label{fig:Chimera}
\end{figure}

\subsection{Minor Embedding Algorithms}

Early work on minor embedding into hardware graphs has focused on complete graphs (or \emph{cliques}). 
A clique $K_n$ is a graph on $n$ vertices in which all vertices are adjacent to one another. 
If $K_n$ can be embedded into a hardware graph $H$, 
then any other graph $G$ with at most $n$ vertices can be embedded in the same hardware graph $H$, 
since $G$ is isomorphic to a subgraph of $K_n$. 
The TRIAD algorithm was the first technique to embed cliques in hardware graphs with 
limited connectivity of the qubits~\citep{Choi2011}. 

The TRIAD algorithm associates each vertex of the problem graph with a chain of vertices of the hardware graph 
that is long enough to have at least one vertex that is adjacent to some vertex of all other chains. 
These 
chains can be embedded into a Chimera graph, 
where a clique $K_{L M}$ fits into $C_{M,M,L}$~\citep{Choi2011}.
In fact, 
it is possible to embed a clique of size $L M +1$, but no clique larger than  $L \left(M+1\right)$, in $C_{M,M,L}$~\citep{klymko2014adiabatic}. 
Later work by \citet{boothby2016fast} generalized the form by which such clique embeddings can be obtained and consequently showed that there is an exponential number of such embeddings in the Chimera graph, 
which can be helpful if some qubits are inoperable and thus some vertices of the hardware graph $H$ are missing.

Figure~\ref{fig:triad_clique} illustrates how $K_{32}$ can be embedded in $C_{8,8,4}$ 
by dividing $32$ vertices into groups of $4$ vertices, which are indexed from $1$ to $8$.  
The first group of 4 vertices is associated with all left partitions of the first column of 
unit cells and also the right partition of the bottom unit cell. The second group of 4 vertices is 
associated with all left partitions of the second column, except the last one, and also with the 
right partitions of the occupied cells in the second row from the bottom. 
Similar L-shaped chains follow for the remaining 6 groups. 
The vertices in distinct groups are adjacent to one another through the cells in the upper 
triangle of the grid, and the vertices within each group are adjacent to one another through the 
cells in the main diagonal of the grid. 
If we associate the remaining cells with a single additional vertex, 
then we can embed $K_{33}$ instead. 
We know from~\cite{klymko2014adiabatic} that we cannot embed $K_{37}$, 
but it is not known if cliques $K_{34}$, $K_{35}$, or $K_{36}$ could be embedded in $C_{8,8,4}$.

\begin{figure}[h!]
\centering
\includegraphics[width=0.3\columnwidth]{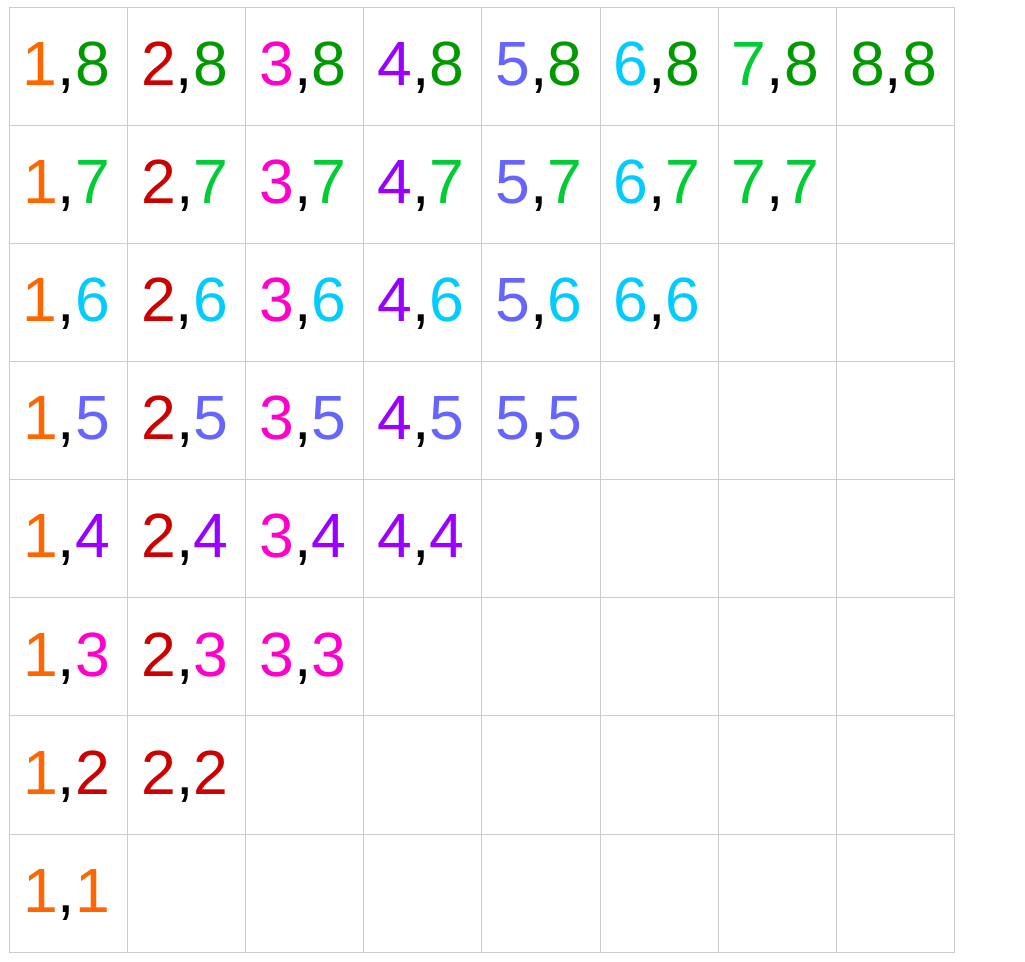} \qquad \includegraphics[width=0.3\columnwidth]{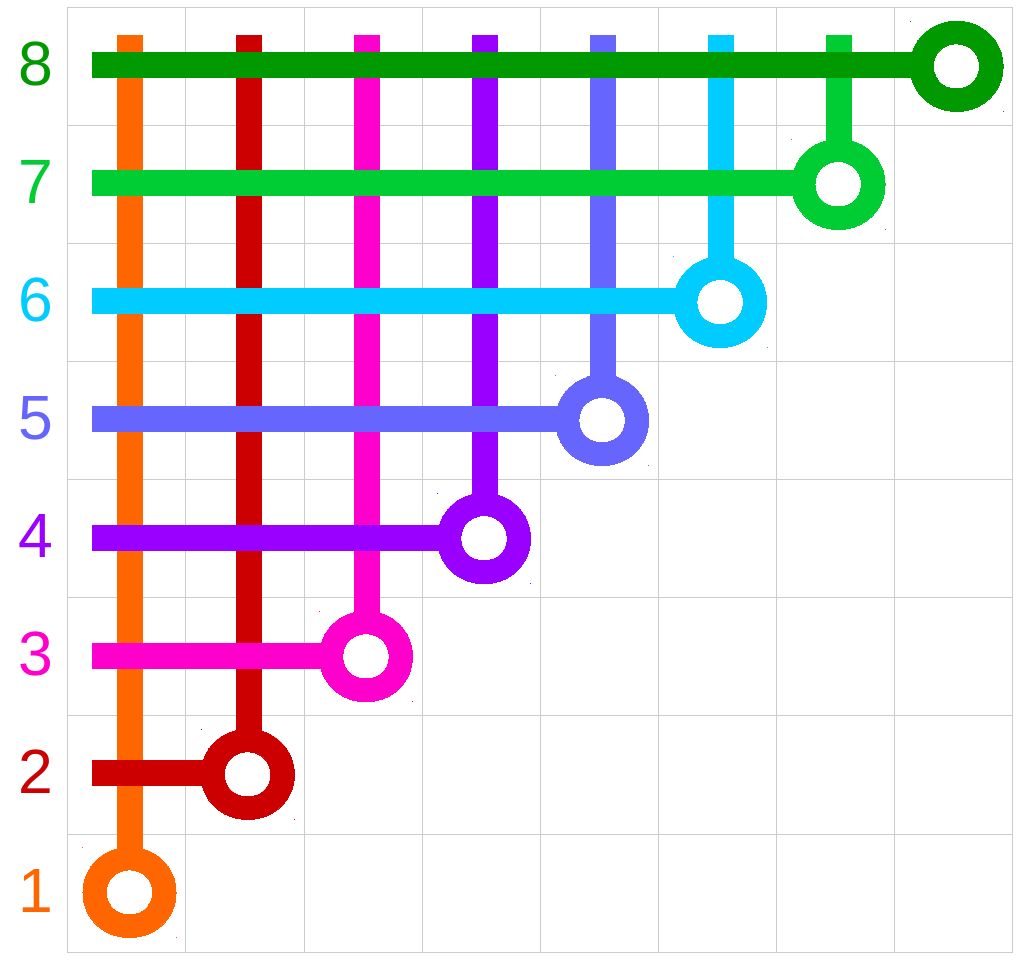}
\caption{An embedding of $K_{32}$ in $C_{8,8,4}$. Each group numbered from 1 to 8 consists of four connected vertices.  
On the left, those numbers are associated with left and right partitions of cells in the grid. On the right, vertical and horizontal lines correspond to each group occupying left and right partitions, and circles for both, following the notation in~\citet{boothby2016fast}.}
\label{fig:triad_clique}
\end{figure}

Recent work by~\cite{Date2019} has focused on limiting the number of qubits used when embedding graphs with at most $M L$ vertices.  
By reducing the number of qubits associated with each variable, their approach is able to obtain QUBO solutions that are closer to the optimal value. 
Another recent line of inquiry concerns embedding the product of graphs, 
which naturally arise as the problem graph of some formulations~\citep{Zaribafiyan2017}.
There are also other general-purpose approaches that break a QUBO problem into smaller parts, 
for example 
by decomposition~\citep{Bian2016} or fixing some variables to their likely value in optimal solutions~\citep{Karimi2017}. 

The line of work that we will explore in this paper consists of embedding problem graphs with more than $M L$ vertices without decompostion, in particular for the case of dense problem graphs.  
In sparse problem graphs, heuristics have been quite successful~\citep{CMR,yang}. 
Among those, one of the most widely used is the CMR algorithm~\citep{CMR}. 
In dense problem graphs, the state-of-the-art consists of using a virtual hardware as an 
intermediary for the embedding. 
The virtual hardware consists of a particular minor of the Chimera graph $C_{M,M,L}$, 
which is chosen to preserve the ability to embed large and dense graphs while making it easy to describe the family of minors that can be obtained from it. 
This idea was pioneered by~\citet{goodrich2018optimizing} 
with a complete bipartite graph $K_{ML,ML}$ as virtual hardware\change{, 
and is currently the state-of-the-art for embedding general problem graphs in QA hardware}. 
Figure~\ref{fig:virtual_hardware} illustrates how $K_{64,64}$ can be embedded in $C_{16,16,4}$: 
each group of 4 vertices is associated with all right partitions of a given row 
or with all left partitions of a given column.

\begin{figure}[h!]
\centering
\includegraphics[width=0.6\columnwidth]{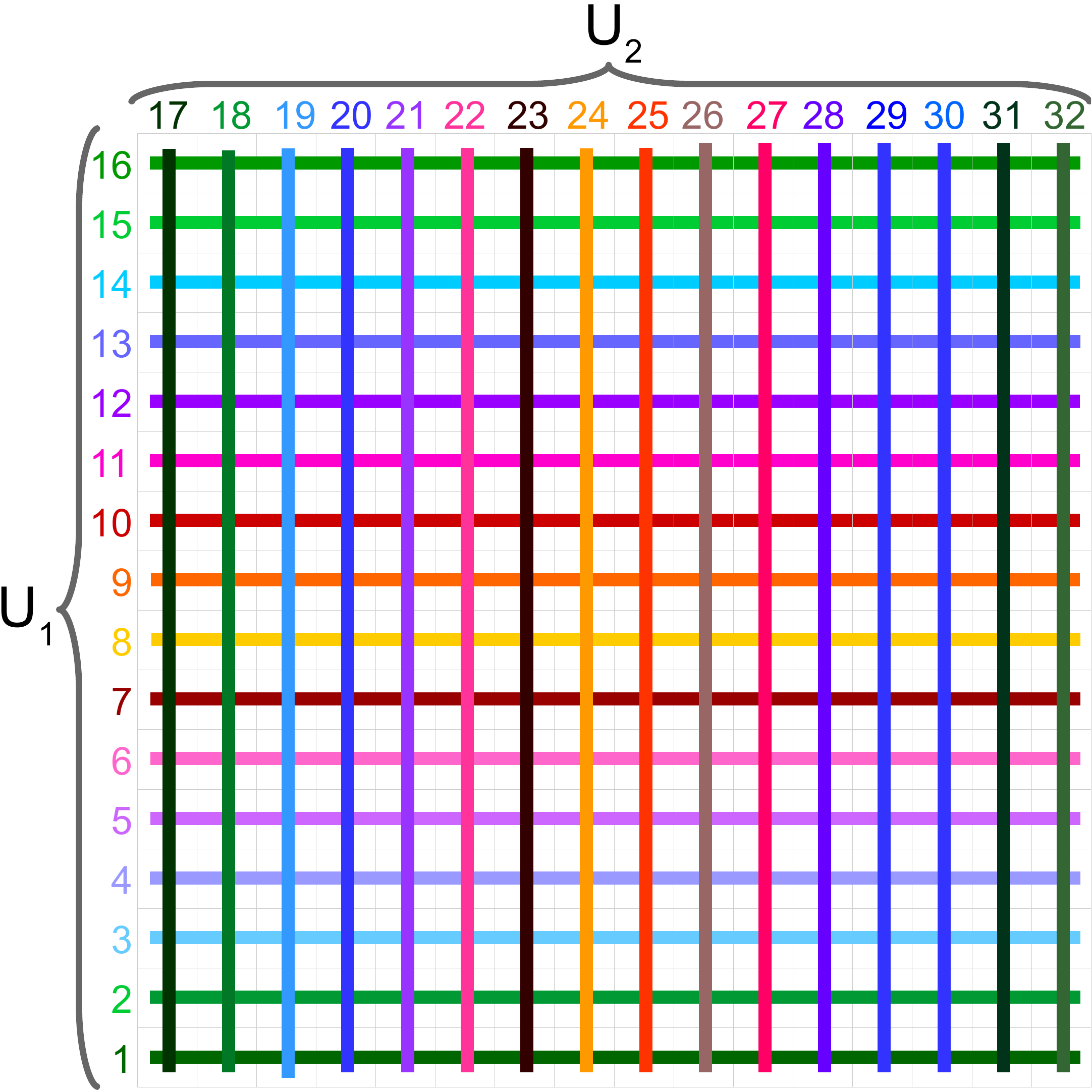}
\caption{An embedding of $K_{64,64}$ in $C_{16,16,4}$ using groups of 4 vertices, 
the first 16 are associated with the right partitions of cells in each row (set $U_1$) and the last 16 with the left partitions of cells in each column (set $U_2$).}
\label{fig:virtual_hardware}
\end{figure}

Any 
minor of $K_{ML,ML}$ 
is isomorphic to a subgraph of one among $ML$ minors of $K_{ML,ML}$~\citep{Hamilton2017}. 
In essence, each vertex $v_i \in V(G)$ is assigned to vertices in either one or both partitions of $H = K_{ML,ML}$ to obtain an embedding of $G$, 
and thus edge $\{u_j, u_k\} \in F(H)$ is contracted if vertex $v_i$ is assigned to both $u_j$ and $u_k$. 

The premise in \citet{goodrich2018optimizing} 
is to assign  vertices of an Odd Cycle Transversal~(OCT) of the problem graph to vertices in both partitions of $K_{ML,ML}$. An OCT of a graph $G$ is a set $T$ 
\change{such that every odd cycle in $G$ has at least one vertex in $T$}, 
hence implying that the removal of $T$ results in a bipartite graph, 
and consequently the remaining vertices are each assigned to a single vertex of $K_{ML,ML}$. 
Those authors  
observed that an OCT of $G$ having minimum size implies the minimum size of a complete bipartite graph in which $G$ can be embedded.
If $T$ is such a minimum size OCT and $V_1$ and $V_2$ are the resulting partitions of the subset of vertices defined by 
$V(G) \setminus T$, then it follows that $G$ can only be embedded in complete bipartite graphs having at least $2 |T| + |V_1| + |V_2|$ vertices.

Later work in \citet{goodrich2018b} used ILP formulations to find an OCT of minimum size, hence minimizing $|T|$ as an approach to determine if a given problem graph is embeddable in $K_{ML,ML}$. They report that customized algorithms can find an OCT of smaller graphs faster. However, they also acknowledge that a general-purpose solver is more effective in cases where the problem graphs are harder to embed. 
\change{We show in this paper that the OCT that embeds $G$ into a complete bipartite graph is not necessarily of minimum size.  This serves as a motivation for considering a different optimization approach}.

\subsection{Related Work}

Some authors have been exploring how to solve a broader class of optimization problems with QA hardware. 
Recent work by \cite{Groebner} uses Groebner bases to represent optimization problems involving polynomial functions of higher order on binary domains as QUBO problems, 
which can then potentially be solved by existing QA hardware. 
Subsequent work by \cite{Graver} uses Graver bases to achieve the same with integer non-linear optimization problems, which may also include constraints.
\change{A number of other quadratizations  that can be applied to such problems is summarized by \cite{Quadratizations}.}

One can also solve a QUBO \change{using} integer linear programming, since the quadratic terms on binary variables can be linearized with an extended formulation~\citep{Padberg1989}. 
In recent work, \cite{Coffrin} uses ILP to verify the solutions generated by QA hardware.

\section{Contributions of This Paper}\label{sec:contr}

We show how Integer Linear Programming~(ILP) can be used as an effective~preprocessing step in AQC, 
especially for problem graphs with more vertices than the largest embeddable cliques. 
Note that we are not interested in solving minor embedding problems that could be nearly as difficult as the corresponding QUBO~problem. 
We focus instead on how classical optimization algorithms could leverage the potential of quantum optimization algorithms. 
Hence, we strive for a balance between computational speed and the ability to embed larger problem graphs by 
defining simple formulations that exploit the structure of Chimera graphs. 
In each of these formulations, we cluster the vertices of a minor of the Chimera graph in some partitions and formulate a problem of deciding how to assign vertices of the problem graph to one or more of such partitions. 
In summary, our main contributions are:
\begin{enumerate}[(i)]
\item We propose Template Embeddings~(TEs) as a generalization of the virtual hardware concept. Each 
template embedding is a minor of the Chimera graph that can embed a variety of problem graphs 
with few edge contractions. 
We study \newchange{two} classes of those: the Bipartite TE~(BTE), as the virtual hardware in~\citet{goodrich2018optimizing}; 
\newchange{as well as the Quadripartite TE~(QTE), as a generalization of the former}.

\item \change{We show that every embedding in BTE is associated with an OCT, but OCTs of minimum
size do not certify that a given problem graph cannot be embedded in 
BTE.}
\item We 
present ILP formulations to determine how to embed a problem graph on the minor of each template embedding with competitive results. For BTE, the formulation provides a certificate of 
embeddability or lack thereof. 
\end{enumerate}

\section{Bipartite Template Embedding}\label{sec:bte}

For a Chimera graph $C_{M,M,L}$, 
BTE consists of the minor $K_{ML,ML}$ used as virtual hardware by \citet{goodrich2018optimizing}, 
in which the vertices of the hardware graph are partitioned into sets $U_1$ and $U_2$ of size $ML$ each.
The construction of BTE is 
described in Figure~\ref{fig:virtual_hardware}.

In order to embed a problem graph $G$ in BTE, 
we need to determine which vertices of $V(G)$ should be assigned to partitions $U_1$ and $U_2$. 
A vertex assigned to a single partition should only be adjacent to vertices assigned to the other partition. 
If assigning all vertices is proven impossible, then $G$ cannot be embedded in BTE.  
If all vertices are assigned to at least one partition, then the solution defines a valid embedding. 

Before formulating the embeddability of a problem graph $G$ in BTE, 
we 
\change{discuss how BTE embeddings relate to OCTs and OCTs of minimum size.} 

\subsection{OCTs and Bipartite Embedding}

In this section, we characterize the relationship between OCTs of a graph $G$ 
and the embedding of $G$ in a complete bipartite graph $K_{m_1,m_2}$. 
More specifically, 
we show that:
\begin{enumerate}[(i)]
\item The set of vertices $S \subseteq V(G)$ that are assigned to both partitions of $K_{m_1,m_2}$ in any embedding of $G$ is a superset of some OCT $T \subset V(G)$, 
but $T$ may be a proper subset of $S$ in each of the possible embeddings.

\item  The largest OCT $T$ contained in $S$ may not be an OCT of minimum size in any of the possible embeddings.
\end{enumerate} 

Those results are shown in the following propositions.

\begin{proposition}\label{prop:supsetOCT}
For any embedding of a graph $G$ in $K_{m_1,m_2}$, 
the set of vertices $S \subseteq V(G)$ assigned to both partitions of $K_{m_1,m_2}$ is such that there is an OCT $T$ of $G$ for which $T \subseteq S$. 
In some cases, $T \subset S$ in every possible embedding.
\end{proposition}
\begin{proof}
First we show that the set of vertices $S$ assigned to both partitions in any embedding is a superset of an OCT $T$. Let us suppose, for contradiction, that none of the vertices incident to an odd cycle of $G$, say $v_1 v_2 \ldots v_k v_1$ for $k \geq 3$ and odd, are in $S$. 
Since each of those vertices is assigned to one partition of $K_{m_1,m_2}$, whereas consecutive vertices should necessarily be in different partitions, 
then one partition is assigned to each $v_i$ with even $i$ 
and the other partition is assigned to each $v_i$ with odd $i$. 
However, since $v_k$ and $v_1$ are adjacent and assigned to the same partition, 
then we do not have a valid embedding, and we reach a contradiction. Hence, $S$ contains 
at least one vertex of every odd cycle in $G$ and is indeed an OCT. This proves the first claim.

Next we show that the embedding of certain graphs implies that we assign to both partitions a vertex that is not incident to any odd cycle. 
In particular, let us consider a \emph{star graph} $K_{1,m}$, 
where a single vertex $v$ is adjacent to all the other vertices and 
consequently the only OCT is an empty set. 
If $m>\max\{m_1,m_2\}$ but $m \leq m_1 + m_2 -2$, then we can only embed $K_{1,m}$ in $K_{m_1,m_2}$ if vertex $v$ is assigned to both partitions, 
in which case each of the remaining $m$ vertices can then be assigned to either partition of $K_{m_1,m_2}$. 
\end{proof}

\change{In other words, 
there is always an OCT in the set of vertices $S$ that should be assigned to both partitions of $K_{m_1,m_2}$. 
However, that set may also contain other vertices. Proposition~\ref{prop:supsetOCT} shows that some of those vertices may not be incident to any OCT. In what follows, we illustrate that the remaining vertices may define an OCT that is not of minimum size.}

\change{
Consider the graph $G$ in Figure~\ref{fig:proof}, where vertices $v_1$ and $v_2$ are adjacent and each is also adjacent to all vertices in set \newchange{$V_A := \{ v_4, v_5, v_6, v_7\}$}, whereas $v_3$ is adjacent to all vertices in sets $V_A $ and \newchange{$V_B := \{ v_{8}, v_{9}, v_{10}, v_{11}\}$}. 
Vertices $v_1$ and $v_2$ are incident to all odd cycles of $G$, such as $v_1 v_4 v_2 v_1$, 
and consequently we only need to remove one of them to obtain a graph with no odd cycles. 
In addition, 
$v_3$ is incident to some odd cycles, such as $v_1 v_4 v_3 v_5 v_2 v_1$. 
Hence, $\{v_1\}$ and $\{v_2\}$ are OCTs of minimum size, but $\{v_1, v_3\}$ is also an OCT. 
}

\begin{figure}[h!]
\centering
\includegraphics[width=0.6\columnwidth]{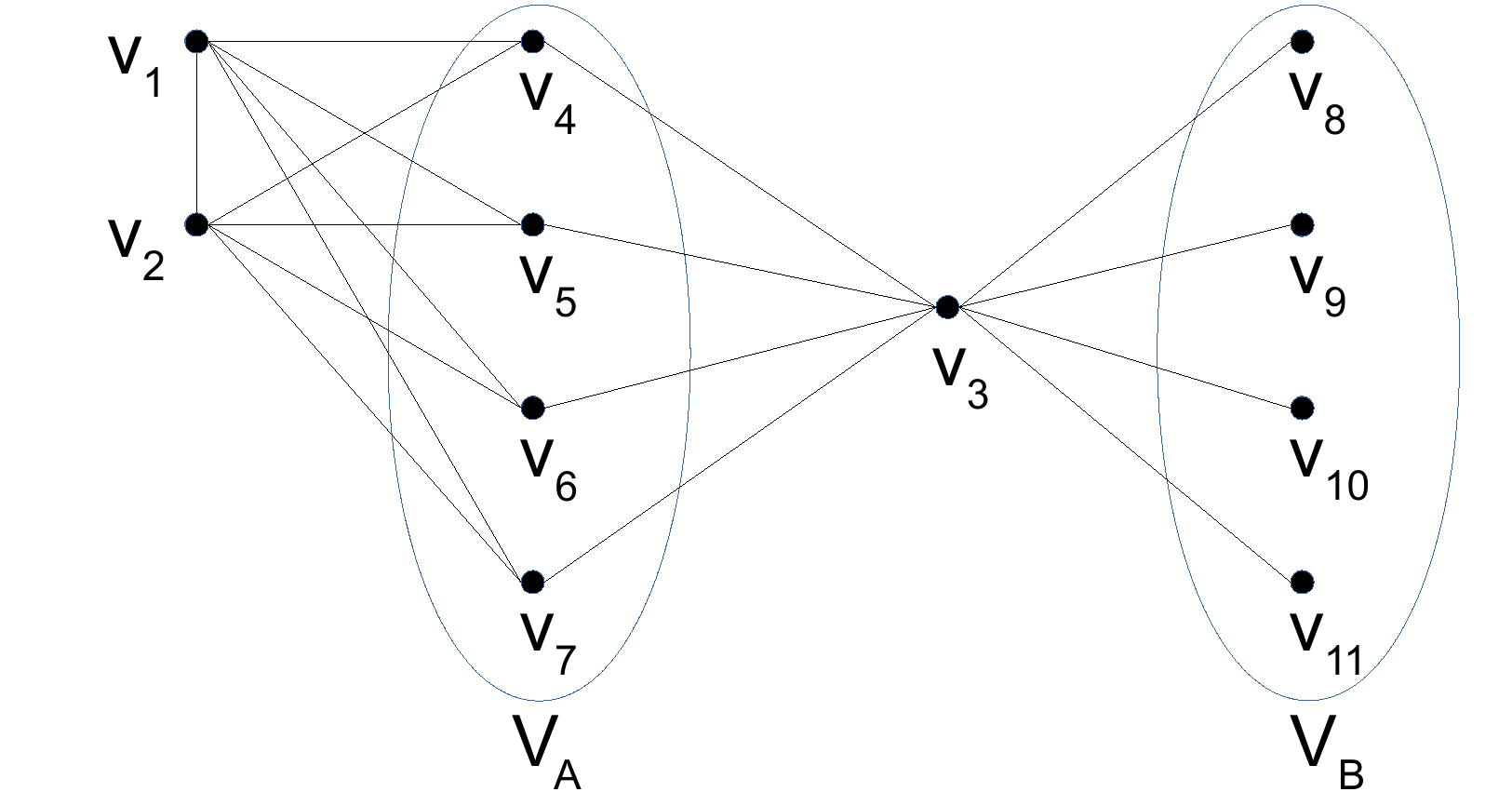}
\caption{\change{Example of a graph $G$ for which assigning vertices in an OCT of minimum size to both partitions of a bipartite graph $H$ may not suffice to obtain a valid embedding of $G$ in $H$.}}
\label{fig:proof}
\end{figure} 

\begin{figure}[h!]
\centering
\includegraphics[width=0.7\columnwidth]{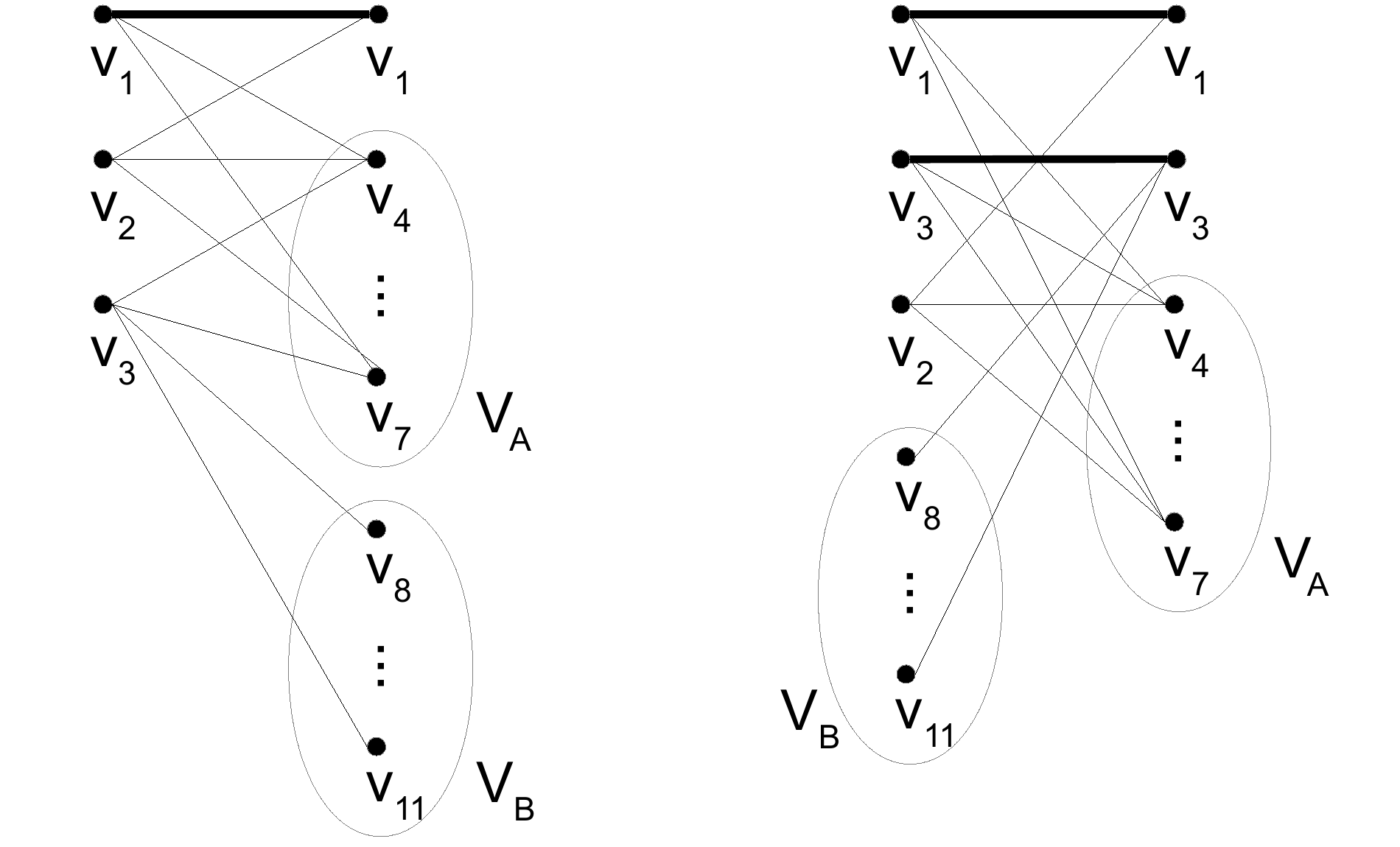}
\caption{\change{Embeddings of graph $G$ from Figure~\ref{fig:proof} in a bipartite graph. On the left, we assign only the vertices of a minimum size OCT $\{v_1\}$ to both partitions and consequently minimize the number of vertices used in both partitions to \newchange{12}. On the right, we assign the vertices of a larger OCT $\{ v_1, v_3\}$ to both partitions and consequently reduce the largest number of vertices assigned to either partition \newchange{from 9 to 7}.}}
\label{fig:embeddings}
\end{figure} 

\change{Let us first consider what happens if we only assign an OCT of minimum size to both partitions of $K_{m_1,m_2}$, which would imply $S = T = \{ v_1 \}$ or $S = T = \{ v_2 \}$. Since $G - v_1$ and $G - v_2$ are isomorphic, we only consider assigning $v_1$ to both partitions and the vertices in $G - v_1$ to a single partition each. 
In that case, vertex $v_2$ and the vertices in set $V_A$ are necessarily in different partitions. Since vertex $v_3$ is adjacent to the vertices in $V_A$, it follows that $v_3$ is in the same partition as $v_2$. Since the vertices in $V_B$ are all adjacent to $v_3$, then the vertices in $V_B$ are in the same partition as those in $V_A$. Consequently, we need to assign $\{ v_2, v_3 \}$ to one partition and $V_A \cup V_B$ to the other partition. 
Therefore, we can only embed graph $G$ in $K_{m_1,m_2}$ by 
assigning the vertices in a minimum size OCT to both partitions 
and the rest to a single partition each if $m_1 \geq 3$ and 
\newchange{$m_2 \geq |V_A|+|V_B|+1 = 9$} or vice-versa. This is illustrated on the left of  Figure~\ref{fig:embeddings}.}

\change{ Now let us consider what happens if we assign a larger OCT to both partitions of $K_{m_1,m_2}$, 
say $\{v_1, v_3\}$.
In the graph $G \setminus \{v_1, v_3\}$, vertex $v_2$ and the vertices in $V_A$ are again in different partitions. Since the vertices in $V_B$ are not adjacent to the other remaining vertices, then one possible partitioning of the remaining vertices is $\{ v_2 \} \cup V_B$ and $V_A$, hence implying that we can embed $G$ in $K_{m_1,m_2}$ if \newchange{$m_1 \geq |V_B|+3 = 7$} and \newchange{$m_2 \geq |V_A|+2 = 6$} or vice-versa. This is illustrated on the right of Figure~\ref{fig:embeddings}.}

\change{If \newchange{$m_1 = m_2 = 8$}, 
then assigning only the vertices of a minimum size OCT $\{ v_1 \}$ to both partitions and 
assigning $V_A \cup V_B$ to a single partition is not feasible since no partition has \newchange{9} vertices.  
However, it is possible to assign only the vertices of the OCT $\{ v_1, v_3 \}$ to both partitions and obtain an embedding, since in such case we only need to assign \newchange{7} vertices to one partition of $K_{m_1,m_2}$ and \newchange{6} vertices to the other partition.}

The rationale for finding an OCT $T$ of minimum size is that it allows $G$ to be 
embedded in the smallest complete bipartite graph, 
regardless of the size of each partition in which we are embedding. 
However, the size of the partitions matter. 
If we denote the partitions of $G \setminus T$ as the sets of vertices $V_1$ and $V_2$, then $G$ can 
only be embedded with such partitions in $K_{m_1,m_2}$ if $\min(m_1,m_2) \geq |T|$ and 
$\max(m_1, m_2) \geq |T| + |V_1| + |V_2|$. 

The graph in Figure~\ref{fig:proof}  
can be embedded with an OCT of minimum size in \newchange{$K_{3,9}$}, which has only \newchange{12} vertices.  
It can also be embedded with a larger OCT in \newchange{$K_{6,7}$}, which has \newchange{13} vertices. 
\newchange{If we were to consider a Chimera graph $C_{16,16,4}$ and tried to embed a similar graph in which $|V_A|=|V_B|=32$, only the second embedding above would be possible in BTE.}

\subsection{Exact Bipartite Embedding}

The formulation below determines if a problem graph $G$ is embeddable in BTE. 

\paragraph{Decision Variables} For each vertex $v_i \in V(G)$ and $k \in \{1,2\}$, 
let $y_{i,k} \in \{0,1\}$ be a binary variable for whether vertex $v_i$ is assigned to partition $U_k$ 
and let $y'_i \in \{0,1\}$ be a binary variable denoting whether $v_i$ is assigned to any partition. 

\paragraph{Objective Function}
The following expression aims at assign as many vertices as possible:
\[
\max \sum_{i=1}^n y'_i.
\]
If the ILP solver reports an upper bound lower than $n$, then it is not possible to embed $G$ in BTE.

\paragraph{Constraints}
For each vertex $v_i \in V(G)$, we associate both types of decision variables as follows:
\begin{equation}\label{eq:any_assignment}
y'_i \leq y_{i,1} + y_{i,2}.
\end{equation}

For each partition $U_k$, $k \in \{1,2\}$, 
no more than $ML$ vertices of $V(G)$ should be assigned to it:
\begin{equation}\label{eq:max_partition}
\sum_{i=1}^n y_{i,k} \leq M L.
\end{equation}

For each edge $\{ v_i, v_j \} \in E(G)$, $i<j$, vertices $v_i$ and $v_j$ should not be assigned to a single and same partition: 
\begin{equation}\label{eq:u1_exclusion}
y_{i,1} + y_{j,1} - y_{i,2} - y_{j,2} \leq 1
\end{equation}
\begin{equation}\label{eq:u2_exclusion}
y_{i,2} + y_{j,2} - y_{i,1} - y_{j,1} \leq 1.
\end{equation}
The constraints above are the canonical cuts on the unit hypercube defined by the binary variables~\citep{CanonicalCuts}. 
Each corresponds to the tightest single inequality on such space because it separates a single combination of values for those variables and is tight for each combination of values that differs in only one variable. 
For example, the first inequality above separates $(y_{i,1}, y_{i,2}, y_{j,1}, y_{j,2}) = (1, 0, 1, 0)$ from the feasible set while holding at equality for the adjacent assignments $(y_{i,1}, y_{i,2}, y_{j,1}, y_{j,2}) = (0, 0, 1, 0), (1, 1, 1, 0), (1, 0, 0, 0),$ and $(1, 0, 1, 1)$.

\begin{proposition}[Certificate of Embeddability]\label{PropEmbed}
Graph $G$ is embeddable in BTE if and only if there is a solution to the 
ILP formulation with objective value $|V(G)|$.
\end{proposition}
\begin{proof}
\change{We start with the premise of assigning vertices to partitions. In order to embed $G$ in BTE, every vertex $v$ of $G$ should be assigned to one or more vertices of BTE. If $v$ is assigned to vertices $u_1$ and $u_2$ in BTE and those vertices are in the same partition, then $u_1$ and $u_2$ have the same neighbors and we can obtain an equivalent embedding by assigning $v$ to only one of them. 
Therefore, 
every embedding of $G$ in BTE can be reduced to an embedding in which each vertex of $G$ is assigned to at most one vertex of each partition and the premise of assigning vertices of $G$ to partitions of BTE is valid.}

\change{Next we show that any graph $G$ that can be embedded in BTE defines a solution to the ILP formulation. 
Based on the discussion above, for any such graph we can define a mapping of every one of its vertices to partitions $U_1$ and $U_2$ of BTE that corresponds to an embedding of $G$ in BTE, 
which implies that $y_{i,k} = 1$ if $v_i \in V(G)$ is assigned to partition $U_k$, $k \in \{1,2\}$, and $y_{i,k} = 0$ otherwise.
By definition, we can infer that three constraints of the ILP formulation are always satisfied with such a mapping. 
First, no more than $ML$ vertices of $G$ are assigned to either $U_1$ or $U_2$ in any embedding in BTE because that is the number of vertices in each of those partitions, 
and thus any such embedding satisfies  constraint~\eqref{eq:max_partition}. 
Second, no pair of adjacent vertices in $G$ are assigned only to the same partition in BTE because the corresponding vertices in BTE would not be adjacent, 
and thus any embedding in BTE would not assign any such pair only to partition $U_1$ as prevented by constraint~\eqref{eq:u1_exclusion} or to partition $U_2$ as prevented by constraint~\eqref{eq:u2_exclusion}. 
In addition, 
since every vertex $v_i \in V(G)$ is assigned to at least one partition, then it follows that $y_{i,1}+y_{i,2} \geq 1$ and thus there is a feasible solution in which $y_i' = 1 ~ \forall v_i \in V(G)$ and the objective function value of $|V(G)|$ is attainable because constraint~\eqref{eq:any_assignment} is not violated by such assignments. 
Therefore, there is a solution of value $|V(G)|$ for every embedding of $G$ in BTE.
}

\change{Finally, we show that any graph $G$ for which we cannot find such a solution to the ILP formulation cannot be embedded in BTE. 
Since a solution with all variables equal to zero is feasible, we just need to consider the case in which the optimal value is less than $|V(G)|$. We will prove this by contradiction. Suppose, for contradiction, 
that there is an optimal solution with value $|V(G)|$ for a graph $G$ that cannot be embedded in BTE. That would imply that each vertex of $G$ is assigned to at least one partition due to constraint~\eqref{eq:any_assignment} and optimality.  Further, no more than $ML$ vertices are assigned to each partition due to satisfaction of constraint~\eqref{eq:max_partition}, and adjacent vertices are not assigned solely to vertices in $U_1$ due to satisfaction of constraint~\eqref{eq:u1_exclusion} or to vertices in $U_2$ due to satisfaction of constraint~\eqref{eq:u2_exclusion}. 
Consequently, such a solution would contradictorily define a valid embedding of $G$. Therefore, the value of any solution to the ILP formulation associated with a graph $G$ that is not embeddable in BTE is less than $|V(G)|$.
}
\end{proof}

An important implication of Proposition~\ref{PropEmbed} is that we can terminate the search process 
and conclude that $G$ is not embeddable in BTE once the search 
determines that there is no solution with objective value $|V(G)|$.  
Search algorithms 
incorporate a number of domain reduction techniques, 
which 
lead to considerable pruning of the search space, 
and this can be favorably exploited to produce a certificate of 
non-embeddability.  Furthermore, any solution with objective value $|V(G)|$ 
is an optimal solution, and the search can be terminated if one of those is found.

\section{Quadripartite Template Embedding}\label{sec:qte}

We define the Quadripartite Template Embedding~(QTE) as a minor of the Chimera graph with vertices partitioned into sets $U_1$, $U_2$, $U_3$, and $U_4$. 
Each vertex in $U_1$ is adjacent to all vertices in $U_2$, those in $U_2$ are also each adjacent to a distinct vertex in $U_3$, and those in $U_3$ are also adjacent to all vertices in $U_4$. In other words, the subgraph induced on $U_1 \cup U_2$ and $U_3 \cup U_4$ are both complete bipartite graphs, 
and the subgraph induced in $U_2 \cup U_3$ is a perfect matching. 
In fact, embedding on QTE generalizes embedding on BTE, since BTE is a minor of QTE after contracting all edges between sets $U_2$ and $U_3$: 
$U_1 \cup U_4$ defines one partition of BTE in that case and the other partition is defined by the vertices resulting from contracting the edges between $U_2$ and $U_3$. 
For a Chimera graph $C_{M,M,L}$ with $M$ even and thus $P:=M/2$ integer, the size of those partitions are: $|U_1|=PL$, $|U_2|=ML$, $|U_3|=ML$, and $|U_4|=PL$.
Figure~\ref{fig:qte_B} illustrates that minor of $C_{16,16,4}$.

\change{The set of minors of QTE is a superset of the minors of BTE. QTE's minors include some larger graphs with fewer adjacencies that would not be possible to embed in BTE. In fact, BTE is a minor of QTE in which the edges between vertices in partitions $U_2$ and $U_3$ are contracted. QTE offers the flexibility that two adjacent vertices in $U_2$ and $U_3$ can be mapped to distinct vertices in $G$. }

\begin{figure}[h!]
\centering
\includegraphics[width=0.6\columnwidth]{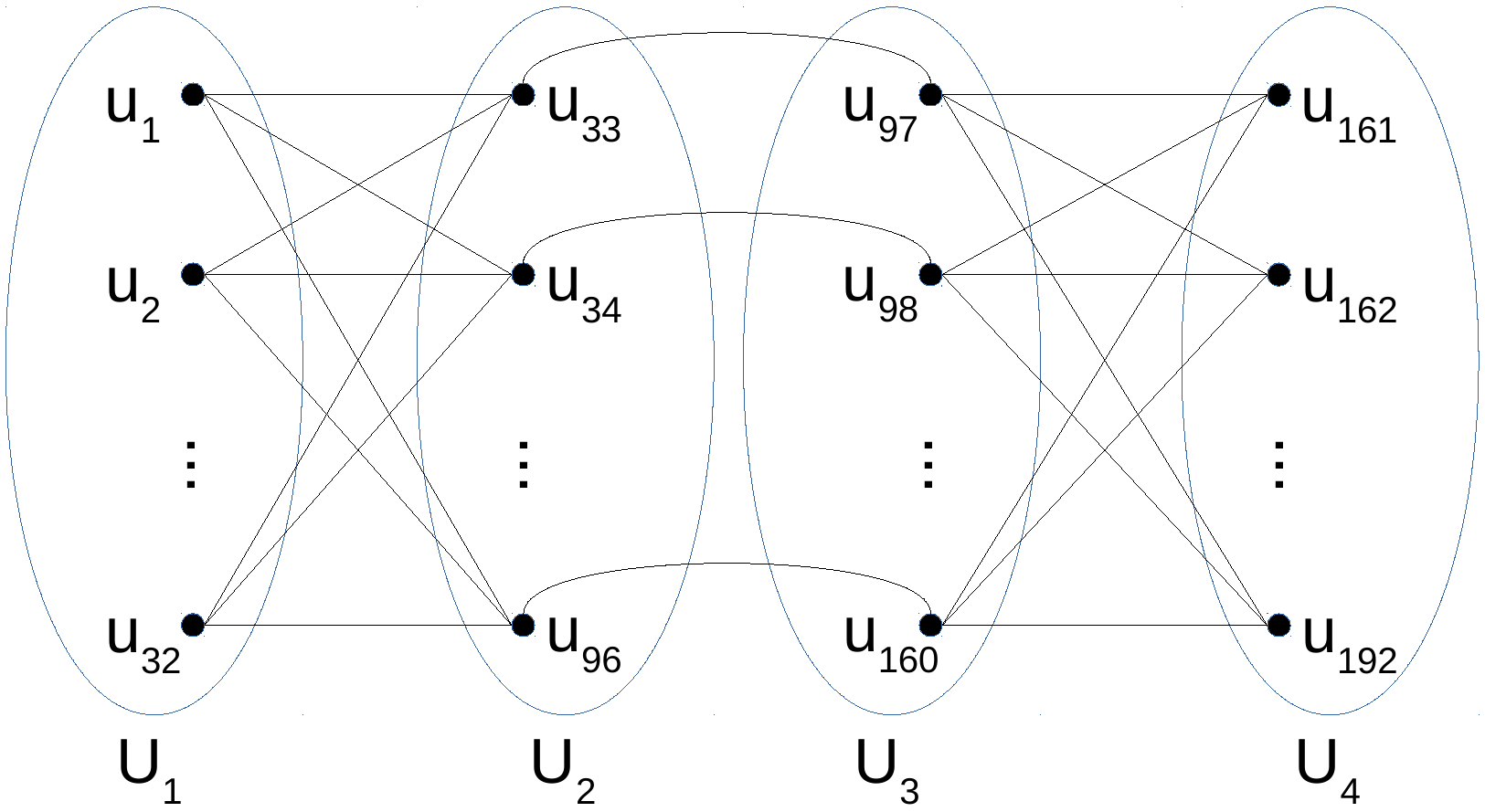}
\caption{QTE minor of Chimera graph $C_{16,16,4}$.}
\label{fig:qte_B}
\end{figure} 

QTE can be obtained from $C_{M,M,L}$ as follows. 
Set $U_1$ consists of $P$ groups of $L$ vertices, 
each group obtained by contracting the $L$ vertices of the right partitions of one of the top $P$ rows of the $M \times M$ grid. 
Set $U_2$ consists of $M$ groups of $L$ vertices, 
each group obtained by contracting the $L$ vertices of the left partitions in the top $P$ rows of one of the $M$ columns of the grid. 
Set $U_3$ consists of $M$ groups of $L$  vertices, 
each group obtained by contracting the $L$ vertices of the left partitions in the bottom $P$ rows of one of the $M$ columns of the grid. 
Set $U_4$ consists of $P$ groups of $L$ vertices, 
each group obtained by contracting the $L$ vertices of the right partitions of one of the bottom $P$ rows of the grid. 
Figure~\ref{fig:qte_A} illustrates QTE in $C_{16,16,4}$.

\begin{figure}[h!]
\centering
\includegraphics[width=0.6\columnwidth]{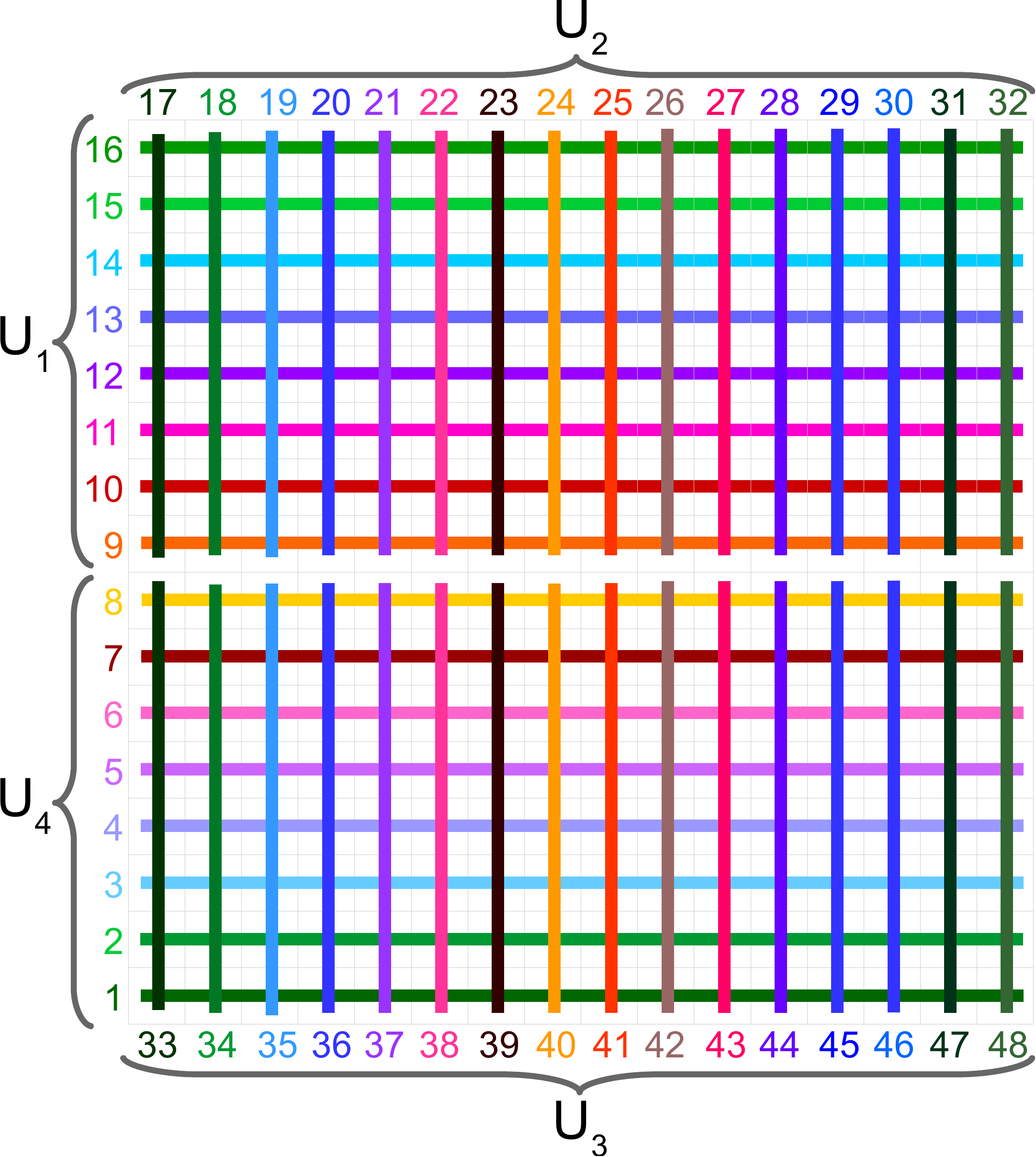}
\caption{
QTE in $C_{16,16,4}$, defining partition $U_1$ with the top 8 horizontal groups, $U_2$ with the top 16 vertical groups,  $U_3$ with the bottom 16 vertical groups, and $U_4$ with the bottom 8 horizontal groups. Each group has 4 vertices. }
\label{fig:qte_A}
\end{figure}

In order to embed a problem graph in QTE, 
each vertex should be assigned to a sequence of adjacent partitions 
and each pair of adjacent vertices $v_i$ and $v_j$ should be assigned to vertices $u_i$ and $u_j$ that are adjacent in QTE. 
\change{For simplicity, 
we only consider that $v_i$ and $v_j$ have been assigned to vertices $u_i$ and $u_j$ that are adjacent in QTE if those vertices are in distinct partitions that together induce a complete bipartite graph}, 
i.e., $U_1$ and $U_2$ or $U_3$ and $U_4$. 
Hence, 
we ignore the possibility of assuming $u_i$ and $u_j$ adjacent if one of these vertices is in partition $U_2$ and the other vertex is in partition $U_3$. 
Otherwise, we would need to explicitly assign the vertices of the problem graph to specific vertices in those partitions 
instead of merely deciding that the vertices of the problem graph are assigned to some vertex in the partition, 
which would make the formulation considerably more complex. 
\change{More specifically, we could potentially decide the particular vertex of the harware graph $H$ that is associated to each vertex assigned to sets $U_2$ and $U_3$ in order to leverage the fact that each vertex of $U_2$ is adjacent to a vertex of $U_3$. However, the number of decision variables would be substantially larger in that case. } 
Given our aim for simplicity, 
the formulation below does not provide a certificate of embeddability in QTE \change{because we only exploit adjacencies between $U_2$ and $U_3$ if a same vertex is assigned to both}.

\paragraph{Decision Variables} For each vertex $v_i \in V(G)$ and $k \in \{1,2,3,4\}$, 
let $y_{i,k} \in \{0,1\}$ be a binary variable for whether vertex $v_i$ is assigned to partition $U_k$,  
and let $y'_i \in \{0,1\}$ be a binary variable denoting whether $v_i$ is assigned to any partition. 
For each edge $\{v_i, v_j\} \in E(G)$, assuming $i < j$, let $z_{i,j}^k \in \{0,1\}$ be an 
auxiliary binary variable implying that the adjacency between vertices $v_i$ and $v_j$ is ensured 
by assigning vertex $v_i$ to partition $U_k$ and vertex $j$ to the partition in which all vertices 
are adjacent to those in $U_k$.

\paragraph{Objective Function} We maximize vertices assigned:
\[
\max \sum_{i=1}^n y'_i.
\]

\paragraph{Constraints}
The first constraint associates the first two types of variables for each vertex $v_i \in V(G)$, as in BTE:
\[
y'_i \leq \sum_{k=1}^4 y_{i,k}.
\]

For each partition $U_k$, $k \in \{1,2,3,4\}$, 
the number of vertices assigned to $U_k$ is bounded by the size of that partition:
\[
\sum_{i = 1}^n y_{i,k} \leq |U_k|.
\]

For each vertex $v_i \in V(G)$, 
we want the set of partitions to which $v_i$ is assigned to be pairwise contiguous. 
We formulate that with constraints preventing each possible discontinuity: 
(i) assigning $v_i$ to $U_1$ and $U_3$ implies that $v_i$ is also assigned to $U_2$; 
(ii) assigning $v_i$ to $U_1$ and $U_4$ implies that $v_i$ is also assigned to $U_2$ and $U_3$; 
and (iii) assigning $v_i$ to $U_2$ and $U_4$ implies that $v_i$ is also assigned to $U_3$. 
Hence, we use canonical cuts on the unit hypercube in the corresponding subspaces to exclude the assignments $(y_{i,1}, y_{i,2}, y_{i,3}) = (1,0,1)$ for (i); $(y_{i,1}, y_{i,2}, y_{i,4)}) = (1,0,1)$ and $(y_{i,1}, y_{i,3}, y_{i,4)}) = (1,0,1)$ for (ii); and $(y_{i,2}, y_{i,3}, y_{i,4}) = (1,0,1)$ for (iii):
\[
y_{i,1} + y_{i,3} - y_{i,2} \leq 1
\]
\[
y_{i,1} + y_{i,4} - y_{i,2} \leq 1
\]
\[
y_{i,1} + y_{i,4} - y_{i,3} \leq 1
\]
\[
y_{i,2} + y_{i,4} - y_{i,3} \leq 1.
\]
For (ii), 
it would suffice to exclude $(y_{i,1}, y_{i,2}, y_{i,3}, y_{i,4}) = (1,0,0,1)$ with $y_{i,1} + y_{i,4} - y_{i,2} - y_{i,3} \leq 1$ because the other cases are covered. 
However, by summing the two inequalities used for (ii)  we obtain the implied inequality $2 y_{i,1} + 2 y_{i,4} - y_{i,2} - y_{i,3} \leq 2$, which is stronger than $y_{i,1} + y_{i,4} - y_{i,2} - y_{i,3} \leq 1$ on continuous domains in $[0,1]$ and excludes additional fractional values such as $(y_{i,1}, y_{i,2}, y_{i,3}, y_{i,4}) = (1,0.5,0.5,1)$. 
Note that a formulation that has a smaller feasible set when the binary variables on $\{0,1\}$ are relaxed to continuous variables on $[0,1]$ is considered as stronger and often is solved faster. 

For each edge $\{v_i, v_j\} \in E(G)$, $i < j$, we want vertex $v_i$ assigned to at least one partition $U_k$ such that vertex $v_j$ is assigned to the corresponding partition $U_l$ where the set of vertices $U_k \cup U_l$ induces a complete bipartite graph. 
In other words, we want $v_i$ and $v_j$ respectively assigned to either (i) $U_1$ and $U_2$; (ii) $U_2$ and $U_1$; (iii) $U_3$ and $U_4$; or (iv) $U_4$ and $U_3$:
\[
y_{i,1} \geq z_{i,j}^1,\qquad
y_{j,2} \geq z_{i,j}^1
\]
\[
y_{i,2} \geq z_{i,j}^2,\qquad
y_{j,1} \geq z_{i,j}^2
\]
\[
y_{i,3} \geq z_{i,j}^3,\qquad
y_{j,4} \geq z_{i,j}^3
\]
\[
y_{i,4} \geq z_{i,j}^4,\qquad
y_{j,3} \geq z_{i,j}^4
\]
\[
\sum_{k=1}^4 z_{i,j}^k \geq 1.
\]
Note that the first two inequalities above imply $y_{i,1} + y_{j,2} \geq 2 z_{i,j}^1$, which alone would also be a valid formulation for (i), 
and the same follows for conditions (ii), (iii), and (iv). 
However, having a pair of inequalities instead makes the formulation stronger since it excludes fractional solutions such as $(y_{i,1}, y_{j,2}, z_{i,j}^1) = (0.5,0.5,1)$. 
Those constraints also imply that each vertex incident to at least one edge should be assigned to a partition, 
and thus an non-embeddable problem graph may lead to an infeasible solution.

\section{Experiments}\label{sec:exps}

We implemented the ILP formulations for the template embeddings in $C_{16,16,4}$ \change{and $C_{20,20,4}$} using Gurobi \change{9}.0.0. 
We compare our results with those obtained with Fast-OCT-Reduce~(FOR) using the source code from~\citet{goodrich2018optimizing}, \change{ 
which is the state-of-the-art for embedding general graphs in QA hardware}. 
\change{
In other words, we test on the hardware graph currently commercialized by D-Wave Systems ($C_{16,16,4}$) and a larger hardware graph ($C_{20,20,4}$) to compare the scalability of both approaches. 
}
We use five random generators of graphs with a density parameter $p=.25$ for Low density, $p=.5$ for Medium density, and $p=.75$ for High density. 
Four generators are from~\citet{goodrich2018optimizing}: Barab\'asi-Albert, Erd\H{o}s-R\'{e}nyi, Regular, and Noisy Bipartite. 
We implement Percolation based on long-range percolation graphs~\citep{Percolation}.
For each vertex $v_i$, we draw a random number $\chi_i \in [0,1)$ and we include edge $\{v_i,v_j\}$ with probability $\min \left\{ 1, \frac{p}{|\chi_i - \chi_j|} \right\}$. 
The problem graphs used in the experiments and a summary of the results can be downloaded 
from \url{ftp://ftp.merl.com/pub/raghunathan/TemplateEmbedding-TestSet/}. 
The source code can be downloaded from \url{https://www.merl.com/research/license/TEAQC}.

Since TRIAD can easily generate a clique of size 64 for $C_{16,16,4}$, 
we generate 5 random graphs from each generator and with each density for number of vertices ranging from 65 to 128. 
\change{Likewise, TRIAD can easily generate a clique of size 80 for $C_{20,20,4}$, and for that case the number of vertices ranges from 81 to 160 in the experiments. 
\newchange{
The maximum sizes for the tested graphs come from assuming that each vertex would be assigned to a single partition. 
In the case of low and medium density, the expected number of edges of the random graphs is smaller than the number of edges of the template embeddings. 
In the case of high density, the expected number of edges would exceed $\binom{ML}{2}$ if the number of vertices is sufficiently large. 
Hence, we have limited the experiments with random graphs having high density to at most 105 vertices in the case of $C_{16,16,4}$ and at most 131 vertices in the case of $C_{20,20,4}$. 
In total, we tried to embed 4,225 graphs in $C_{16,16,4}$ and 5,275 graphs in $C_{20,20,4}$.
}
}
The time required to embed problems is currently the 
bottleneck in solving problems on QA hardware. 
\change{With an eye towards cases that would benefit from our classical approach as a presolve for AQC, we set a time limit of 60 seconds 
as a satisficing threshold to identify what graphs can be embedded by each approach. We also complement our analysis with plots that compare runtimes in order to identify the most effective approach.} 
All experiments were conducted on a single thread in an \change{Intel(R) Xeon(R) CPU E5-2640 v3 @ 2.60GHz and
128 GB RAM.}

\newchange{We begin by comparing the total number of graph instances that can be embedded by the different algorithms in Figure~\ref{fig:ttlInst}. For $C_{16,16,4}$, Fast-OCT-Reduce (FOR) can embed 694 instances while BTE and QTE can embed 717 and 670, respectively. Thus, BTE improves on FOR marginally while QTE embeds fewer instances than FOR.  For $C_{20,20,4}$, BTE can embed significantly more instances as compared to FOR (798 vs. 575, respectively).  QTE also does much better than FOR in this case (575 vs. 671).
} 

\newchange{Our next analysis in Figure~\ref{fig:summary} aims to identify the unique strengths of each algorithm beyond the total number of instances that are successfully embedded. We count the instances that are embedded by only one approach or template embedding, i.e. that no other approach was able to embed. From this analysis, we see that every graph that is embedded by FOR is also successfully embedded by either BTE or QTE. Among the template embeddings, BTE embeds more instances than QTE. 
}

\newchange{Tables~\ref{tab:exp_res_c16} and~\ref{tab:exp_res_c20} show the largest graph that each algorithm could respectively embed in $C_{16,16,4}$ and $C_{20,20,4}$  within the time limit for each type of random generator and density, with the largest numbers of each row in bold. This helps to identify the type of random graph generators and densities for which one approach is superior to the the others.
}

\newchange{Table~\ref{tab:each} provides a breakdown of the number of graphs embedded in $C_{16,16,4}$ and $C_{20,20,4}$ by the graph classes considered. This is a dissection of the results in Figure~\ref{fig:ttlInst}, which also complements the results reported in Tables~\ref{tab:exp_res_c16} and~\ref{tab:exp_res_c20} by quantifying the total number of graphs embedded.}

\newchange{Figures~\ref{fig:performance_c16_each} and \ref{fig:performance_c16_any} respectively show the performance profiles of each approach as well as the combination of all template embeddings in comparison with the OCT-based approach in $C_{16,16,4}$. In other words, we see the cumulative number of graphs embedded over time for each case.  Figures~\ref{fig:performance_c20_each} and \ref{fig:performance_c20_any} show the same type of performance profiles with respect to embedding into $C_{20,20,4}$. The performance profiles evidence that the use of template embeddings results in runtimes that are orders of magnitude faster than FOR.}

In some of these figures and tables, we respectively denote $C_{16,16,4}$ and $C_{20,20,4}$ as C16 and C20 for brevity. 

\newchange{One noticeable difference between the template embeddings is whether we can determine if the embedding problem is feasible or not by the time limit. In the case of BTE, we were able to find a feasible embedding or determine that none exists by the time limit. In the case of QTE, it was not possible to determine if an embedding was infeasible by the time limit. That speaks to the importance of having models that are as simple as possible, and also highlights the fact that an ILP formulation can reliably determine if a graph can be embedded in a complete bipartite graph for the sizes that we tested.}

\begin{figure}[h]
\centering
\includegraphics[width=1\columnwidth]{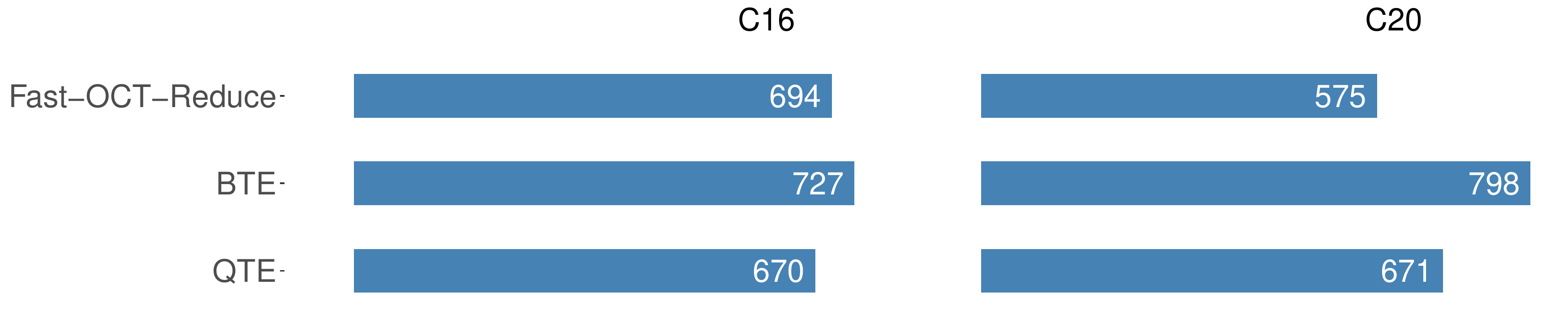}  
\caption{Number of graphs embedded by Fast-OCT-Reduce and each template embedding in C16 and C20. 
}
\label{fig:ttlInst}
\end{figure}

\begin{figure}
\centering
\includegraphics[width=0.49\columnwidth]{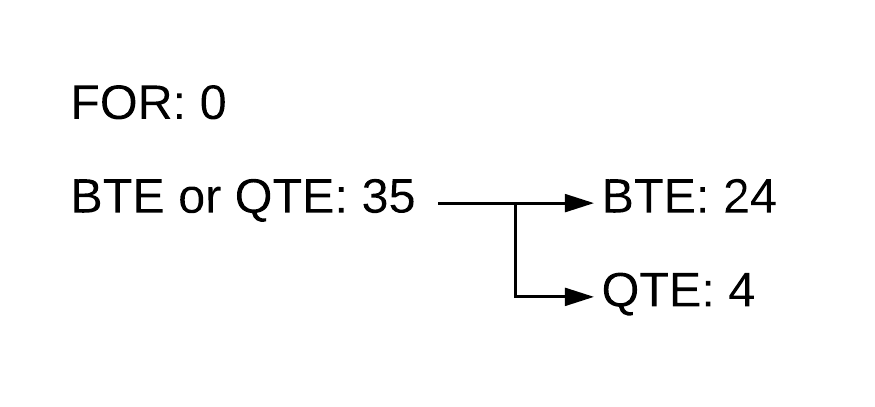}
\includegraphics[width=0.49\columnwidth]{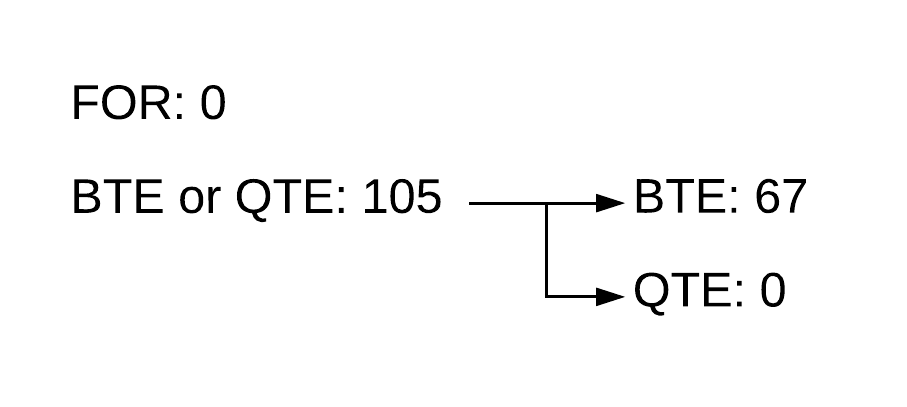}
\caption{Graphs uniquely embedded by each approach or a specific template embedding in C16 (left) and C20 (right). 
}
\label{fig:summary}
\end{figure}

\begin{table}
\centering
\caption{The largest graph embedded by each algorithm per generator and density in C16. 
}\smallskip
\begin{tabular}{@{\extracolsep{0.5pt}}cccccccccccc}
Random Graph & Density & FOR & BTE & QTE \\\hline
\noalign{\vskip2.5pt}
\multirow{3}{*}{Percolation}            
& low     & 67  & \textbf{68}  & \textbf{68}  \\
& medium  & \textbf{66}  & \textbf{66}  & \textbf{66}  \\
& high    & \textbf{66}  & \textbf{66}  & \textbf{66}  \\\hline
\multirow{3}{*}{Barab\'asi-Albert}        
& low     & 77  & \textbf{78}  & 76  \\
& medium  & \textbf{72}  & \textbf{72}  & 71  \\
& high    & \textbf{69}  & \textbf{69}  & \textbf{69} \\\hline
\multirow{3}{*}{Erd\H{o}s-R\'{e}nyi}     
& low     & 79  & \textbf{80}  & 77  \\
& medium  & 71  & \textbf{72}  & \textbf{72}  \\
& high    & \textbf{69}  & \textbf{69}  & 68  \\\hline
\multirow{3}{*}{Regular}                
& low     & \textbf{78}  & \textbf{78}  & 77  \\
& medium  & \textbf{72}  & \textbf{72}  & \textbf{72}  \\
& high    & \textbf{69}  & \textbf{69}  & \textbf{69}  \\\hline
\multirow{3}{*}{Noisy Bipartite} 
& low     & \textbf{92}  & \textbf{92}  & \textbf{92}  \\
& medium  & \textbf{85}  & \textbf{85}  & 82  \\
& high    & \textbf{81}  & \textbf{81}  & 80  \\\hline
\end{tabular}

\label{tab:exp_res_c16}
\end{table}

\begin{table}
\centering
\caption{The largest graph embedded by each algorithm per generator and density in C20. 
}\smallskip
\begin{tabular}{@{\extracolsep{0.5pt}}cccccccccccc}
Random Graph & Density & FOR & BTE & QTE \\\hline
\noalign{\vskip2.5pt}
\multirow{3}{*}{Percolation}            
& low     & 81  & \textbf{84}  & 83  \\
& medium  & \textbf{82}  & \textbf{82}  & \textbf{82}  \\
& high    & \textbf{82}  & \textbf{82}  & \textbf{82}  \\\hline
\multirow{3}{*}{Barab\'asi-Albert}        
& low     & 85  & \textbf{95}  & 92  \\
& medium  & 86  & \textbf{89}  & 87  \\
& high    & \textbf{85}  & \textbf{85}  & 84  \\\hline
\multirow{3}{*}{Erd\H{o}s-R\'{e}nyi}     
& low     & \textbf{97}  & \textbf{97}  & 94  \\
& medium  & \textbf{89}  & \textbf{89}  & 88  \\
& high    & 85  & \textbf{86}  & 85  \\\hline
\multirow{3}{*}{Regular}                
& low     & \textbf{95}  & \textbf{95}  & 94  \\
& medium  & \textbf{88}  & \textbf{88}  & 87  \\
& high    & \textbf{85}  & \textbf{85}  & 84  \\\hline
\multirow{3}{*}{Noisy Bipartite} 
& low     & 110 & \textbf{111} & \textbf{111} \\
& medium  & 100 & \textbf{102} & 99  \\
& high    & 99  & \textbf{100} & 96  \\\hline
\end{tabular}

\label{tab:exp_res_c20}
\end{table}

\begin{table}[h]
\centering
\caption{Number of graphs embedded by Fast-OCT-Reduce~(FOR) and Template Embeddings~(TE) in C16 and C20. 
}
\begin{tabular}{@{\extracolsep{4pt}}lcccc}
& \multicolumn{2}{c}{C16} & \multicolumn{2}{c}{C20}\\
\cline{2-3}
\cline{4-5}
Random Graph    & FOR & TE & FOR & TE \\
\cline{1-1}
\cline{2-2}
\cline{3-3}
\cline{4-4}
\cline{5-5}
\noalign{\vskip2.5pt}
Percolation                    & 35 & 37  & 24 & 38 \\
Barab\'asi-Albert                & 114 & 126 & 67 & 138\\
Erd\H{o}s-R\'{e}nyi           & 128  & 139 & 90 & 150\\
Regular                        & 127 & 129 & 112 & 140\\
 Noisy Bipartite         & 290 & 301 & 282 & 332\\
\cline{1-1}
\cline{2-2}
\cline{3-3}
\cline{4-4}
\cline{5-5}
\noalign{\vskip2.5pt}
Total & 694 & 732 & 575 & 798 \\
\end{tabular}
\label{tab:each}
\end{table}

\begin{figure}[h!]
\centering
\includegraphics[width=0.9\columnwidth]{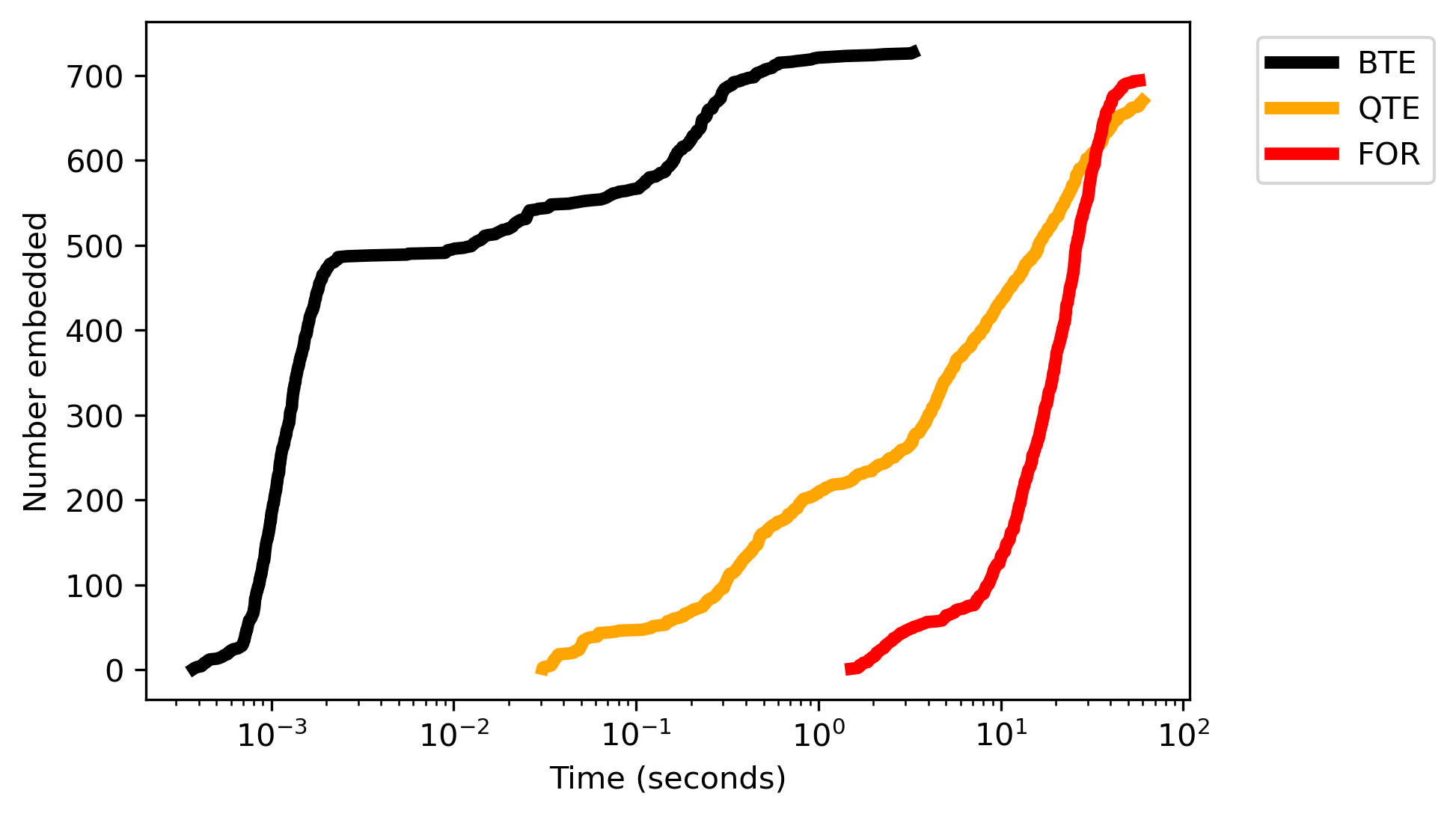} 
\caption{\change{Performance profile of graphs embedded by Fast-OCT-Reduce and each template embedding in C16.} 
}
\label{fig:performance_c16_each}
\end{figure}

\begin{figure}[h!]
\centering
\includegraphics[width=0.9\columnwidth]{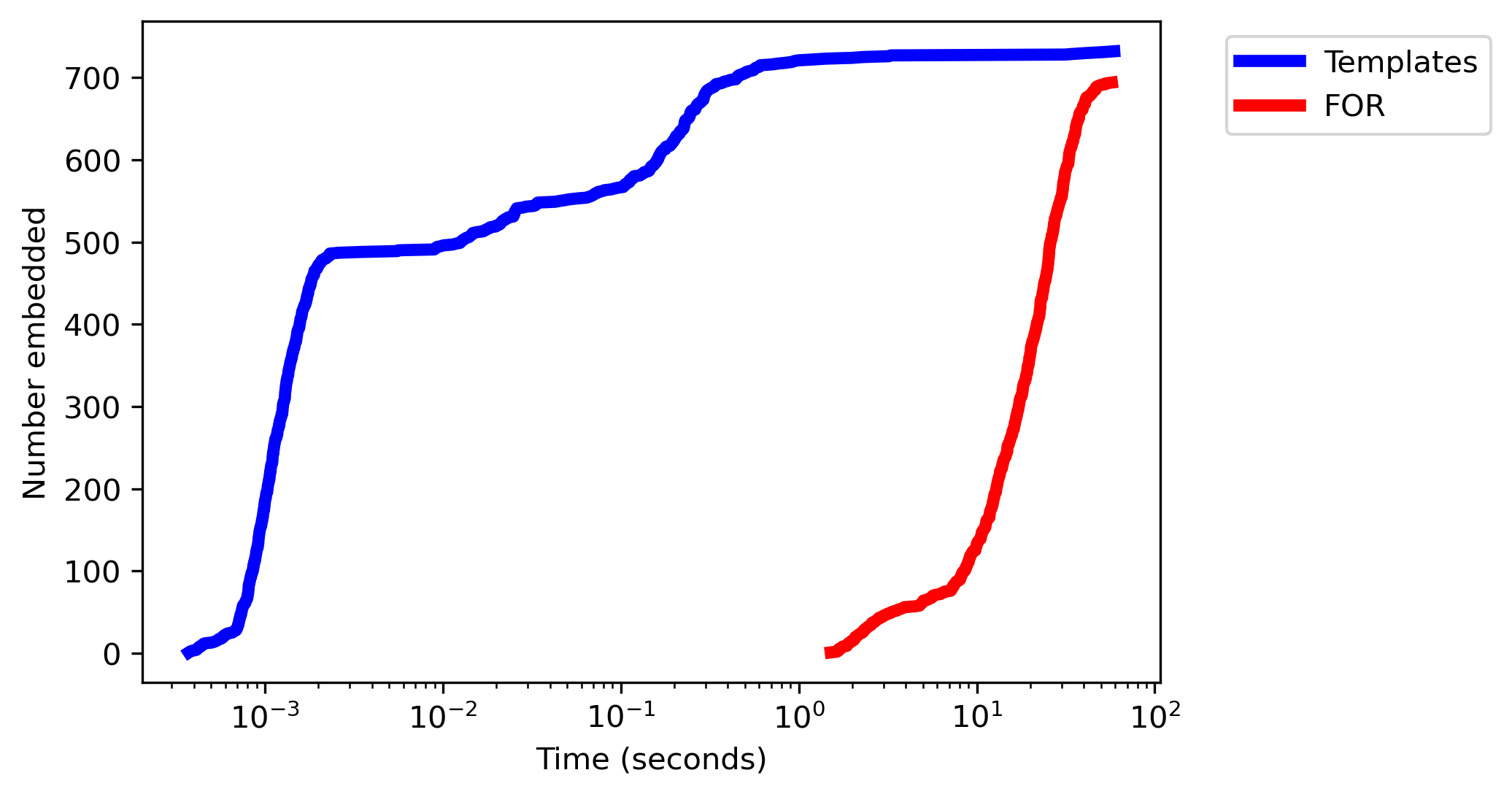} 
\caption{\change{Performance profile of graphs embedded by Fast-OCT-Reduce and any template embedding in C16.} 
}
\label{fig:performance_c16_any}
\end{figure}

\begin{figure}[h!]
\centering
\includegraphics[width=0.9\columnwidth]{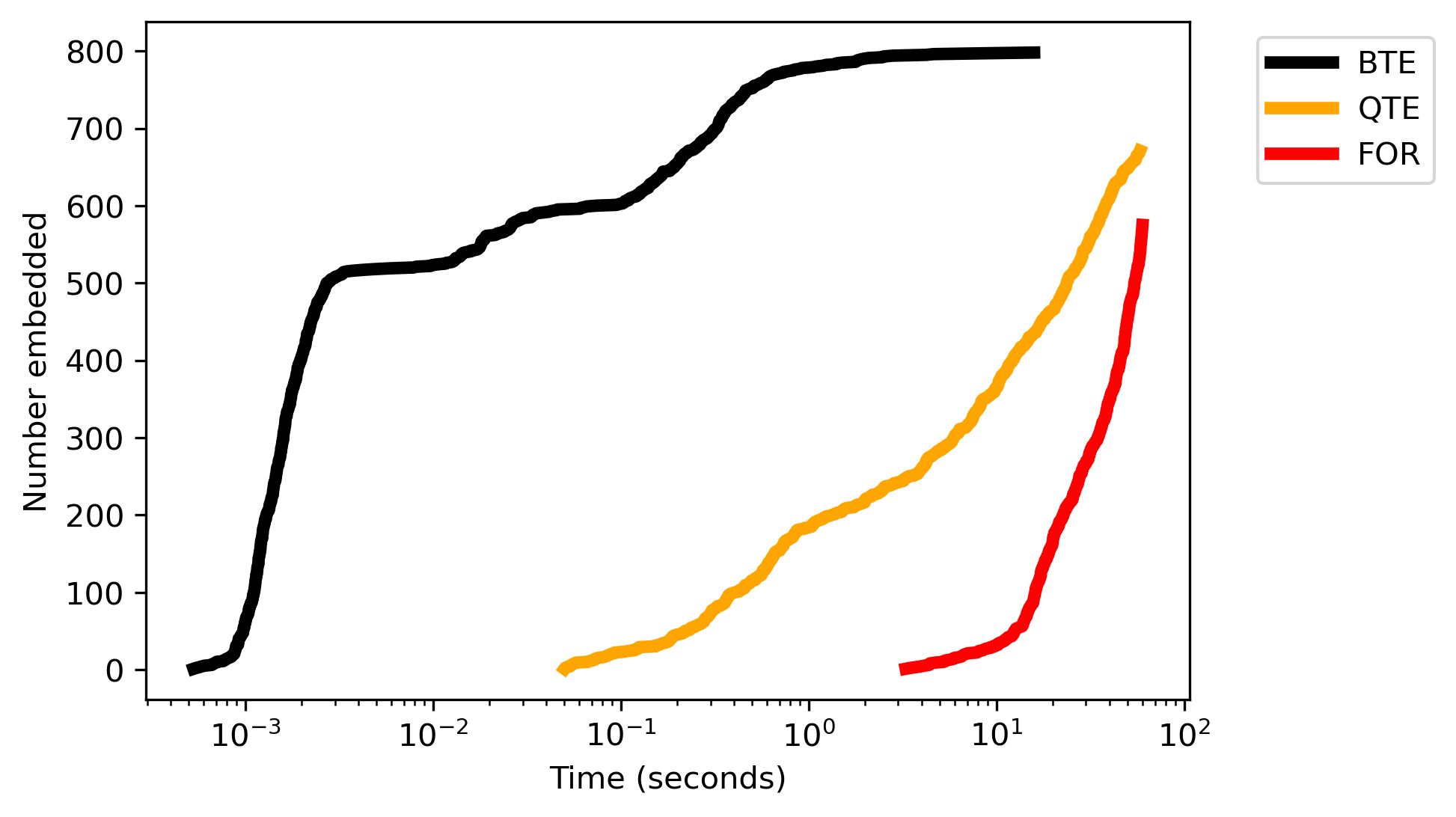} 
\caption{\change{Performance profile of graphs embedded by Fast-OCT-Reduce and each template embedding in C20.} 
}
\label{fig:performance_c20_each}
\end{figure}

\begin{figure}[h!]
\centering
\includegraphics[width=0.9\columnwidth]{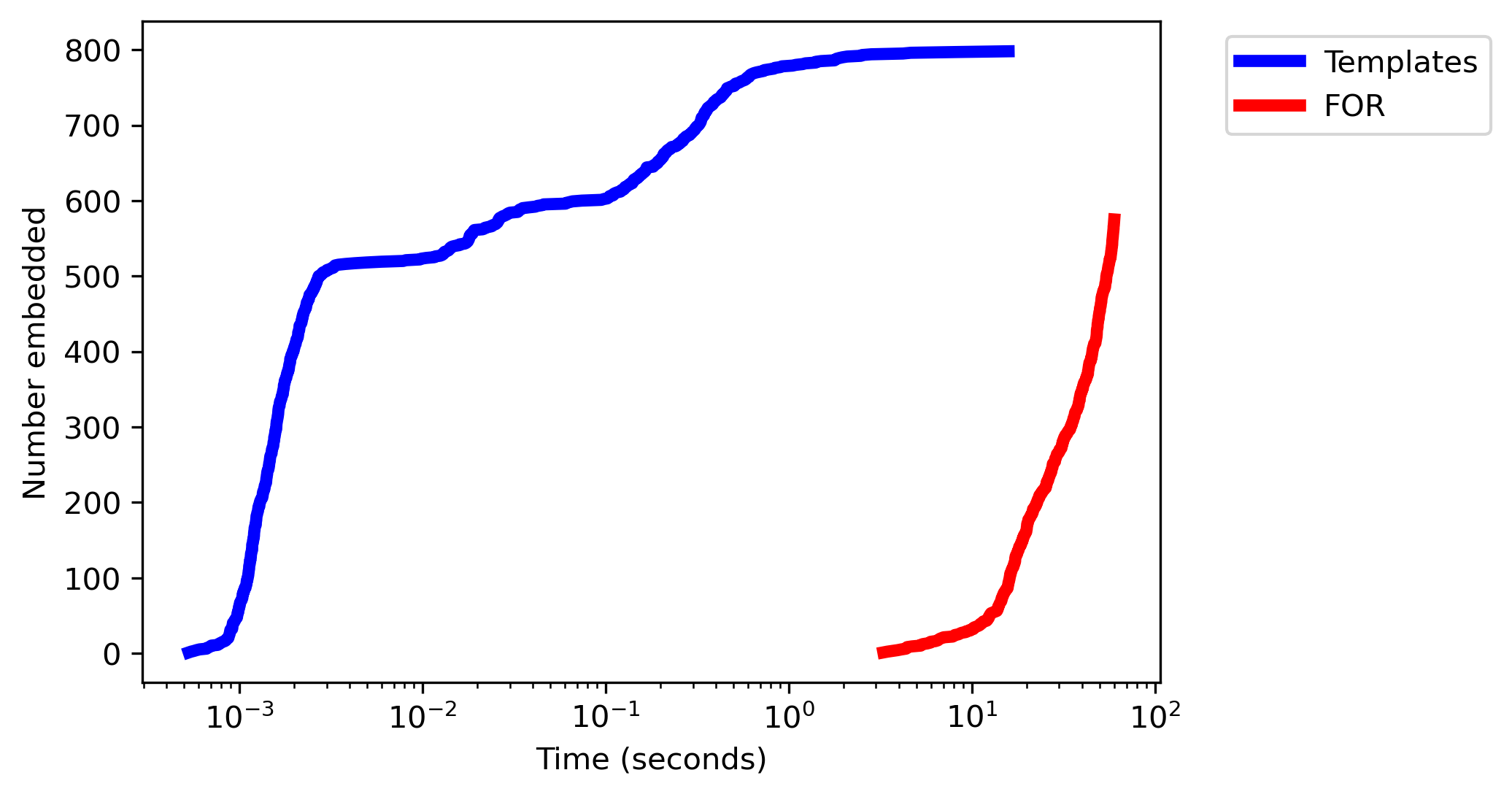} 
\caption{\change{Performance profile of graphs embedded by Fast-OCT-Reduce and any template embedding in C20.} 
}
\label{fig:performance_c20_any}
\end{figure}

\subsection{Analysis} 
The ILP formulations based on template embeddings can  
embed \change{5\% and 39\% more graphs than the OCT-based approach in $C_{16,16,4}$ and $C_{20,20,4}$, respectively. 
More importantly, these graphs can be embedded substantially faster. }  Furthermore, every 
graph embedded by the OCT-based approach is also embedded by one of the 
template embeddings.  
\change{
Out of the 30 different combinations of random graph, density, and hardware graph, BTE embeds the largest graph in all 30 of them. Meanwhile, FOR embeds the largest graph in 19, and QTE in 11. The results are particularly favorable with the larger hardware graph, in which we can often embed more graphs with template embeddings and fewer with the OCT-based approach. That evidences the scalability of the proposed approach. While BTE is a special case of QTE, 
the simpler formulation allowed BTE to embed more Barab\'asi-Albert, Erd\H{o}s-R\'{e}nyi, regular, and noisy bipartite graphs within the time limit in both hardware graphs.  
QTE performed comparatively better with Percolation graphs of lower density. While BTE dominates the performance profile curves, we note that QTE helped embedding extra graphs as the runtime increased. }

\change{Our main takeaway is that solving a simpler embedding problem yields better results. Because of its relatively small formulation, BTE led to very fast and scalable results. 
In addition, as conjectured, QTE showed favorable results with lower density graphs. 
}

\section{Conclusion and Future Work}\label{sec:concl}

We proposed the concept of template embeddings to map quadratic unconstrained binary optimization problems into quantum annealers. 
Each template embedding corresponds to a minor of the Chimera graph that can embed a variety of large and dense graphs. 
We also introduced integer linear programming formulations to find such mappings and showed that one of these formulations is exact, 
and thus certify if a given graph can be embedded in the corresponding \change{minor}. 
Experimental results clearly demonstrate the potential of the  
proposed approach, especially to embed problems having more variables than the maximum embeddable clique in the Chimera graph corresponding to the QA hardware.

Interestingly, our approach makes a better use of quantum annealers  
by leveraging classical optimization algorithms as a preprocessing step. 
The performance of solvers for mixed-integer linear programming has improved by orders of magnitude in the last decades~\citep{Bixby}. 
We believe that there is potential for further coordination between classical and quantum algorithms for discrete optimization.

In future work, we intend to investigate template embeddings 
that are adaptive to problem sparsity and  
incorporate knowledge of faulty qubits in the formulations. 
\change{We should also consider improvements to the formulation of template embeddings and solving strategies to scale up this approach as the size of the hardware graph increases in future releases.}

\change{Another line of work consists of extending  
the template-based approach to minors of Pegasus, the proposed hardware graph for future QA hardware by 
D-Wave Systems~\citep{Pegasus}. However, we note that the contributions of this work are also valid for Pegasus, as observed in a technical report by D-Wave~\citep{DWaveTR}:}

\begin{quote}
\change{
\emph{Much of the embedding support that exists for the Chimera topology may be extended to the Pegasus family of topologies with relative ease. Known constructions for Chimera embeddings  of  structured  problems  translate  to  Pegasus  without  modification,  because Chimera occurs as a subgraph of Pegasus.}}
\end{quote}

\paragraph{Acknowledgements} The first two authors of the paper, Thiago Serra and Teng Huang, were employed by Mitsubishi Electric Research Laboratories during the initial development of this project. They are the corresponding authors. In addition, the authors would like to thank Claudio Leonardo Lucchesi and Yoshiharu Kohayakawa for their input on random graphs.

\bibliographystyle{plainnat}
\bibliography{embedding}

\end{document}